\documentclass[11pt]{article}
\usepackage{./qc}

\usepackage{dsfont}
\usepackage{authblk}
\usepackage[margin=2.5cm]{geometry}

\usepackage{xcolor}

\usepackage{nicematrix}
\usepackage{tikz}
\usetikzlibrary{tikzmark}
\usetikzlibrary{decorations.pathreplacing,calligraphy}

\usepackage{graphicx} 

\newcommand{\shifttext}[2]{%
    \tikz[remember picture,overlay]\node[inner sep=0pt] (shifted) at (#1) {#2};%
}

\usepackage{cleveref}
\crefname{problem}{problem}{problems}
\crefname{claim}{claim}{claims}

\usepackage[backend=biber, style=alphabetic, backref=true, hyperref=true, maxbibnames=99]{biblatex}
\usepackage[babel,english=british]{csquotes}
\DefineBibliographyStrings{english}{%
    backrefpage  = {cited on p.}, 
    backrefpages = {cited on pp.} 
}
\definecolor{burgundy}{RGB}{128, 0, 32}

\newcommand{\sing}{\mathrm{sing}}
\newcommand{\type}{\mathrm{type}}

\title{Testing classical properties from quantum data}

\author[1,2]{Matthias C. Caro}
\author[3]{Preksha Naik}
\author[3]{Joseph Slote}
\affil[1]{\small{
University of Warwick, Coventry, UK}}
\affil[2]{\small{Freie Universität Berlin, Berlin, Germany}}
\affil[3]{California Institute of Technology, Pasadena, CA, USA}
\date{}
\begin{document}

\maketitle
\vspace{-3em}
\begin{abstract}
    Many classes of Boolean functions can be tested much faster than they can be learned.
    However, this speedup tends to rely on \textit{query} access to the function $f$.
    When access is limited to random samples $(x,f(x))$---the \textit{passive} testing model and a natural setting for data science---testing can become much harder.
    Here we introduce \textit{quantum passive testing} as a quantum version of this ``data science scenario'': quantum algorithms that test properties of a function $f$ solely from quantum data in the form of copies of the function state $\ket{f}\propto \sum_x\ket{x,f(x)}$.
    Just like classical samples, function states are independent of the property of interest and can be collected well in advance.
    
    \smallskip
    \textbf{Quantum advantage in testing from data: an emerging theme.}
    For three well-established properties---monotonicity, symmetry, and triangle-freeness---we show passive quantum testers are unboundedly- or super-polynomially better than their classical passive testing counterparts, and in fact are competitive with classic \emph{query}-based testers in each case.
    Existing quantum testers for $k$-juntas and linearity can be interpreted as passive quantum testers too and exhibit the same phenomena.

    \smallskip
    \textbf{Inadequacy of Fourier sampling.}
    Our new testers use techniques beyond quantum Fourier sampling, and it turns out this is necessary: we show a certain class of bent functions can be tested from $\mc O(1)$ function states but has a sample complexity lower bound of $2^{\Omega(\sqrt{n})}$ for any tester relying exclusively on Fourier and classical samples.
    
    \smallskip
    \textbf{Classical queries vs. quantum data.}
    Our passive quantum testers are competitive with classical \textit{query}-based testers, but this isn't universal: we exhibit a testing problem that can be solved from $\mathcal{O}(1)$ classical queries but requires $\Omega(2^{n/2})$ function state copies.
    The \textsc{Forrelation} problem provides a separation of the same magnitude in the opposite direction, so we conclude that quantum data and classical queries are ``maximally incomparable'' resources for testing.
    
    \smallskip
    \textbf{Towards lower bounds.} We also begin the study of \emph{lower bounds} for testing from quantum data.
    For quantum monotonicity testing, we prove that the ensembles of \cite{Goldreich2000,black_nearly_2023}, which give exponential lower bounds for classical sample-based testing, do not yield any nontrivial lower bounds for testing from quantum data.
    New insights specific to quantum data will be required for proving copy complexity lower bounds for testing in this model.
\end{abstract}

\newpage

\tableofcontents

\newpage

\section{Introduction}\label{section:introduction}

In property testing we consider a subset $\mc P$ of the set of all Boolean functions $f:\{0,1\}^n\to\{0,1\}$ and aim to find fast algorithms for deciding (with high probability) whether an unknown function $f$ has property $\mc P$ or is $\epsilon$-far from having property $\mc P$; that is, we wish to decide between
\[\text{Case (\textit{i})}\quad f\in \mc P\quad\qquad \text{or}\quad\qquad \text{Case (\textit{ii})}\quad\min_{g\in \mc P}\|f-g\|_1\geq \epsilon\,,\]
promised one of these is the case.
Here $\|f-g\|_1=\Pr_{x\sim \{0,1\}^n}[f(x)\neq g(x)]$ is the $L^1$ distance.
Property testing began in the context of program checking \cite{blum1990self, rubinfeld1996robust}, where it was shown that only $\mc O(1)$ queries to
$f$ are needed to determine (with high probability) whether $f$ is linear or is $\Theta(1)$-far from linear---which compares very favorably to the $\Omega(n)$ query lower bound for \emph{learning} linear functions.
The extreme query efficiency of property testing algorithms soon after played a critical role in interactive proofs and PCP theorems \cite{arora1998probabilistic, arora1998proof, dinur2007pcp}.
Since then property testing has developed into a rich landscape of access models, complexity regimes, and separations \cite{fischer2004art,Rubinfeld2007,DBLP:journals/fttcs/Ron09,DBLP:conf/propertytesting/Sudan10,Goldreich_2017}.

One of the promises of this broad view of property testing, identified very early on \cite{goldreich1998property}, is its potential in data analysis and machine learning: one could run inexpensive property testing algorithms to guide the choice of which long-running 
learning algorithm to use.
But there is an unfortunate catch: the dramatic complexity advantage of testing over learning typically disappears in the natural access model for data analysis and machine learning, where fresh queries to $f$ cannot be made and only a limited dataset $\{(x_j,f(x_j)\}_j$ of random samples from $f$ is available.
This setting is known as \emph{passive} or \emph{sample-based testing} \cite{goldreich1998property}.

\begin{table}[t]
    \centering
    \makebox[\textwidth][c]{
    \renewcommand{\arraystretch}{2}
    \begin{tabular}{ |l|c|c|c|c|c|} 
     \cline{2-6}
     \multicolumn{1}{c|}{}&  \multicolumn{2}{c|}{Quantum}&  \multicolumn{3}{c|}{Classical}\\
      \hline
     Property & Queries & Examples & Samples & Queries & \makecell{Learning\\\footnotesize(from queries)}\\
     \hline
     
     $k$-Juntas &
     \makecell{$\widetilde{\mc O}(\sqrt{k})$\\\shifttext{1.5ex,1ex}{\tiny\cite{ambainis_efficient_2015}}} &
     \makecell{$\mathcal{O}(k)$\\\shifttext{3.5ex,1ex}{\tiny\cite{atici2007quantum}}} &
     \makecell{$\Omega(2^{k/2}+k\log n)$\\\shifttext{6ex,1ex}{\tiny\cite{alon2016active}}} &
     \makecell{$\widetilde{\Theta}(k)$ \\\shifttext{2.1ex,1ex}{\tiny\cite{blais_testing_2009,chockler_lower_2004}}} & \makecell{$\Omega(2^{k}+k\log n)$\\\shifttext{5.3ex,1ex}{\tiny\cite{alon2016active}}}\\
     
     Linearity &
     \makecell{$\mathcal{O}(1)$\\\phantom{\tiny[]}}&
     \makecell{$\Theta(1)$ \\\shifttext{3.6ex,1ex}{\tiny\cite{bernstein1997complexity}}}&
     \makecell{$n+\Theta(1)$ \\\shifttext{6ex,1ex}{\tiny\cite{alon2016active}}}&
     \makecell{$\Theta(1)$\\\shifttext{4ex,1ex}{\tiny\cite{blum1990self}}}& \makecell{$n+\Theta(1)$ \\\shifttext{5.3ex,1ex}{\tiny\cite{alon2016active}}}\\
    
     $\F_2$ degree-$d$ &
     \makecell{$\mathcal{O}(2^d)$\\\phantom{\tiny[]}} &
     \makecell{$\mathcal{O}(n^{d-1})$\\\shifttext{3ex,1ex}{\tiny \cite{arunachalam2023optimal}}}&
     \makecell{$\Theta(n^d)$\\\shifttext{6.1ex,1ex}{\tiny\cite{alon2016active}}} &
     \makecell{$\Theta(2^d)$\\\phantom{\tiny\cite{alon2003testing, bhattacharyya2010optimal}}\shifttext{-4.8ex,1ex}{\tiny\cite{alon2003testing, bhattacharyya2010optimal}}} &
     \makecell{$\Theta(n^d)$\\\shifttext{5.3ex,1ex}{\tiny\cite{alon2016active}}}\\
     
     Monotonicity &
     \makecell{$\widetilde{\mc O}(n^{1/4})$ \\\shifttext{2.5ex,1ex}{\tiny\cite{Belovs2015}}}&
     \makecell{$\tilde{\mc O}(n^2)$ \\ 
     \shifttext{2.0ex,1ex}{\tiny[\Cref{theorem:quantum-monotonicity-testing}]}}&
     \makecell{$2^{\Omega(\sqrt{n})}$\\\shifttext{6.8ex,1ex}{\tiny\cite{black_nearly_2023}}}&
     \makecell{$\widetilde{\mc O}(\sqrt{n})$\\\shifttext{3.9ex,1ex}{\tiny\cite{DBLP:journals/siamcomp/KhotMS18}}} &
     \makecell{$2^{\Omega(\sqrt{n})}$\\\shifttext{5.6ex,1ex}{\tiny\cite{DBLP:conf/focs/BlumBL98}}}\\

     Symmetry &
     \makecell{$\mathcal{O}(1)$\\\phantom{\tiny[]}} &
     \makecell{$\mathcal{O}(1)$\\\shifttext{2.0ex,1ex}{\tiny[\Cref{theorem:quantum-symmetry-testing}]}}
     &
     \makecell{$\Theta (n^{1/4})$\\\shifttext{6ex,1ex}{\tiny \cite{alon2016active}}}
     &
     \makecell{$\mc O(1)$\\\shifttext{3.6ex,1ex}{\tiny\cite{blais2015partially}}} &
     \makecell{$\Theta (n^{1/2})$\\\shifttext{5.3ex,1ex}{\tiny \cite{alon2016active}}}\\

     Triangle-freeness &
     \makecell{$\mathcal{O}(1)$\\\phantom{\tiny[]}} &
     \makecell{$\mathcal{O}(1)$\\\shifttext{1.6ex,1ex}{\tiny[\Cref{theorem:quantum-triangle-freeness-testing}]}}
     &
     \makecell{$\Omega(2^{n/3})$\\\shifttext{5ex,1ex}{\tiny \Cref{remark:triangle-freeness-lower-bound}}}
     &
     \makecell{$\mc O(1)$\\\shifttext{3.6ex,1ex}{\tiny\cite{blais2015partially}}} &
     \makecell{--\\ 
     \phantom{\tiny[]}}\\
    \cline{1-2} \cline{3-6}
    \end{tabular}
    }
    \renewcommand{\arraystretch}{1}
    \vspace{1em}
    \caption{\textbf{Upper and lower bounds for testing and learning in various access models.}
    All bounds are given for (a sufficiently small) constant $\varepsilon>0$.
    Bounds that are given without a reference follow trivially from other bounds in the table.
    }
    \label{tab:q-vs-c}
\end{table}

Indeed, many results in passive testing are lower bounds that grow with $n$, unlike the algorithms available in query-based testing: compare among the ``Classical'' columns in \Cref{tab:q-vs-c}.
In fact, Blais and Yoshida \cite{blaisyoshida2019} showed that if a Boolean property can be tested from $\mc O(1)$ random samples, then the property is of a rather restricted kind.\footnote{In particular, such a property is only a function of the conditional expected values $\E_x[f(x)|x\in S_j]$ of $f$ for sets $S_j$ forming a constant-cardinality partition of the hypercube, $S_1\sqcup\cdots\sqcup S_{\mc O(1)}=\{0,1\}^n$.}

\begin{remark}This is not to say that the classical passive testing model is \textit{uninteresting}; there are many exciting positive results for the model, falling under the umbrella of \emph{sublinear} algorithms.
For example, the line of work \cite{fischerlachishvasudev2015, goldreich_sample-based_2016, DBLP:journals/siamcomp/DallAgnolGL23} showed that the existence of certain constant-query testers implies sample-based algorithms with sublinear dependence on $n$.
But passive testers still cannot compete with query-based testing for many important problems, as the lower bounds in Table \ref{tab:q-vs-c} attest.
\end{remark}

How could we recover large testing speedups in the context of passive testing from data?
In the present work we advocate for quantum computing (and ``quantum datasets'') as an answer.
Viewed from the right perspective, early results in quantum complexity theory
actually demonstrate that quantum data---in the form \textit{quantum examples}, or copies of the \emph{function state} $\ket{f}:=2^{-n/2}\sum_x\ket{x,f(x)}$---can sometimes suffice for highly efficient property testing.
For example, the Bernstein-Vazirani algorithm, usually understood as an $\mc O(1)$ quantum \textit{query} algorithm, really only needs $\mc O(1)$ \textit{function states} to test for linearity \cite{bernstein1997complexity} (vs. $\Omega(n)$ classical samples), and the quantum $k$-junta tester of At{\i}c{\i} and Servedio \cite{atici2007quantum} also requires only $\mc O(k)$ quantum examples (\textit{c.f.} the lower bound of $\Omega(2^{k/2} +k\log n)$ classical samples).
The present work seeks to establish \textit{passive quantum testing} as a fundamental model of property testing by making progress on the question:

\begin{quote}
\centering
    \textit{What is the extent of quantum advantage in testing classical properties from data?}
\end{quote}

Before this work it was not fully clear whether quantum data in the form of quantum examples can lead to testing speedups beyond linearity and $k$-junta-like properties (such as low Fourier degree):
both the Bernstein--Vazirani algorithm and the At{\i}c{\i}--Servedio junta tester rely only on \emph{quantum Fourier sampling} \cite{bernstein1997complexity}, a quantum subroutine which, given copies of $\ket{f}$, returns the label $S\subseteq[n]$ of a Fourier character with probability $\widehat{f}(S)^2$.
Despite the success of quantum Fourier sampling, its utility is restricted to properties that are ``plainly legible'' from the Fourier spectrum.\footnote{As an example of a property \textit{not} detectable from the Fourier spectrum, consider the task of testing if $f$ is a quadratic $\F_2$ polynomial.
It is well-known (see, \textit{e.g.}, \cite[Claim 2.4]{hatami2019higher}) that degree-2 $\F_2$ polynomials can have Fourier coefficients with uniformly exponentially-small magnitudes, so Fourier sampling is not directly useful for this task. Our \Cref{inf-thm:FS-not-sufficient-for-mm-new} below serves as another example.} 

In this work we expand the list of properties with efficient passive quantum testers, including one which provably requires a non-Fourier sampling approach.
We also compare the power of quantum data to that of classical queries, finding that they are (essentially) maximally incomparable as resources for testing.
Finally, we begin a study of lower bounds for testing monotonicity from quantum data by showing that the ensembles leading to exponential lower bounds for classical sample-based testing yield no nontrivial lower bounds for quantum data-based testing. 
In the remainder of the introduction we explore each of these points in greater detail.

\begin{remark}[Where might quantum data appear?]
While from the perspective of complexity theory quantum data leads to a natural counterpart to classical passive testing, and demonstrates a ``data-based'' quantum advantage, the reader may still feel it is not entirely natural from a practical or ``physical'' standpoint.
To the contrary, we contend that quantum data may be a useful component of emerging quantum technologies.
We briefly list some scenarios where quantum data may be a natural object.
\begin{itemize}
    \item Suppose a researcher has time-limited query access to a data-generating process, but does not yet know what questions about the process she will eventually ask.
    She may prefer to store data in quantum memory rather than classical, to broaden the range of questions that can be answered \textit{post hoc}.
    \item In high-latency and bandwidth-limited scenarios, back-and-forth (adaptive, query-based) interaction is not feasible, for example in space exploration.
    If a space probe departing Earth shared some entanglement with a ground station, it could later in its journey encode observations into quantum data and teleport the resulting states back to Earth.
    In such a scenario, the advantages of quantum data could lead to significant speedups in research and analysis.
    \item Rather than sharing the source code for a program $f$, a company may prefer to share a quantum data encoding of it as a form of copy protection---provided the function state is sufficient for the intended application.
\end{itemize}
\end{remark}

\subsection{Quantum advantage in testing from data: an emerging theme} \label{subsection:algorithms-for-properties}
Our first contribution is to expand the list of properties exhibiting quantum advantage in testing from data.
Our algorithms work by finding new quantum ways to exploit insights from prior work in classical testing.
See \Cref{section:passive-quantum-testing-upper-bounds} for proofs.

\bigskip
\noindent\textbf{Symmetry testing.}
A Boolean function is \emph{symmetric} if $f(x)= f(y)$ when $x$ is a permutation of $y$.
We confirm that projecting $\ket{f}$ onto the symmetric subspace suffices for an $\mc O(1)$-copy quantum test.
For comparison, classical passive symmetry testing requires $\Omega(n^{1/4})$ samples \cite{alon2016active}.

\bigskip
\noindent\textbf{Monotonicity testing.}
A Boolean function $f$ is \emph{monotone} if $f(x)\leq f(y)$ when $x\prec y$ in the standard partial order $\prec$ on the hypercube.
Monotonicity has been of central importance in the classical property testing literature \cite{Goldreich2000, Belovs2015, DBLP:journals/siamcomp/KhotMS18}.
We give a quantum algorithm that tests monotonicity with $\tilde{\mc  O}(n^2)$ copies of the function state for $f$, in comparison to the lower bound of $2^{\Omega(\sqrt{n})}$ samples for classical passive testing \cite{black_nearly_2023}.

The algorithm appeals to a characterization of monotonicity in terms of the Fourier spectrum of $f$.
In particular, let $\epsilon$ be the $L^1$ distance between a Boolean function $f$ and the set of all monotone functions.
Then we may relate $\epsilon$ to the Fourier spectrum of $f$ via
\[2\epsilon\;\leq\; \mb I[f] - \textstyle\sum_i\widehat{f}(\{i\})\;\leq\; 4\epsilon n\,.\]
Here $\mb I[f]$ is the total influence of $f$ and is equal to the expected size of a subset $S\subseteq[n]$ sampled according to the Fourier distribution of $f$.
$\mb I[f]$ can thus be easily estimated with Fourier sampling, and the Fourier coefficients $\widehat{f}(\{i\})$ estimated with classical samples.
The bounds above follow from a reinterpretation of the ``pair tester'' characterization of monotonicity \cite{Goldreich2000}, which was not originally Fourier-based. 

\bigskip
\noindent\textbf{Triangle-freeness.} A Boolean function $f$ is \emph{triangle-free} if there are no $x,y$ such that $(x,y,x+y)$ form a \emph{triangle}: $f(x)=f(y)=f(x+y)=1$.
We give a passive quantum triangle-freeness tester that uses only $\mathcal{O}(1)$ copies of $\ket{f}$.
This is to be contrasted with the $\Omega(2^{n/3})$ samples required classically.\footnote{This lower bound, the proof of which we outline in \Cref{subsection;passive-quantum-trianglefreeness-testing}, arises from the requirement of seeing a linearly dependent triple $(x,y,x+y)$ among the sampled inputs.}

It is known that to test triangle-freeness, it suffices to estimate the probability that $(x,y,x+y)$ forms a triangle for uniformly random $x,y$ \cite{fox2011new, hatami2016arithmetic}.
Our test estimates this probability by repeating the following subroutine.
First, measuring copies of $\ket{f}$ in the computational basis allows us to find a uniformly random $y\in f^{-1}(1)$.
Then by measuring the output register of copies of $\ket{f}$, we obtain copies of the entire $1$-preimage state
\[\Ket{f^{-1}(1)}\propto\sum_{x\in\{0,1\}^n,\,f(x)=1}\ket{x}\,.\]
Applying the unitary transformation $U_y\ket{x}=\ket{x+y}$ then allows us to transform copies of $\Ket{f^{-1}(1)}$ into copies of 
\[\ket{f^{-1}(1)+y}\propto \sum_{x\in\{0,1\}^n,\,f(x+y)=1}\ket{x}\, .\]
The overlap $|\braket{f^{-1}(1)|f^{-1}(1)+y}|=\Pr_{x\sim\{0,1\}^n}[f(x)=f(x+y)=1]$ can now be estimated with a SWAP test \cite{buhrman2001quantum}.


\subsection{Fourier sampling does not suffice}
Given that Fourier sampling is sufficient to test linearity \cite{bernstein1997complexity}, $k$-juntas \cite{atici2007quantum, ambainis_efficient_2015}, and (as shown above) monotonicity, one might wonder whether Fourier sampling is ``all that quantum data is good for'' in the context of property testing Boolean functions. 
To the contrary, we exhibit a property for which a Fourier sampling-based approach requires super-polynomially more data than the optimal passive quantum tester.

\begin{theorem}\label{inf-thm:FS-not-sufficient-for-mm-new}
    There is a property $\mathcal{P}$ of Boolean functions on $2n$ bits such that:
    \begin{itemize}
        \item[(i)] There is no algorithm for testing $\mathcal{P}$ that uses $2^{o(\sqrt{n})}$ classical samples and any number of Fourier samples.
        \item[(ii)] There is an efficient quantum algorithm for testing $\mathcal{P}$ from $\mathcal{O}(1)$ copies of $\ket{f}$.
    \end{itemize}
\end{theorem}

\Cref{inf-thm:FS-not-sufficient-for-mm-new} is proved in \Cref{section:FS-not-enough-for-mm} as \Cref{thm:fourier-sampling-not-sufficient-for-mm,thm:quantum-mm-tester}.
The property $\mathcal{P}$ is the \emph{Maiorana-McFarland (MM)} class of bent functions, which take the form $f(x,y)=\langle x,y\rangle + h(x)$ for $h$ any $n$-bit Boolean function (see \textit{e.g.,} \cite{carlet2016four} for more).

To prove (\emph{i}), we show a special subset $F_\text{yes}$ of MM functions with far-from-constant $h$ are indistinguishable from the set $F_\text{no}$ of their ``duals,'' defined by replacing $h(x)$ with $h(y)$.
Every function in both these sets is \textit{bent}---\textit{i.e.}, all Fourier coefficients have equal magnitude---so Fourier samples cannot help.
It thus suffices to lower bound the number of classical samples needed to solve the distinguishing problem.
The set $F_\text{yes}$ is chosen so that for a uniformly-random $\langle x,y\rangle + h(x)$ from $F_\text{yes}$, the distribution of truth tables of $h$ is $2^{c\sqrt{n}}$-wise independent.
This means that for any number of samples less than $2^{c\sqrt{n}}$, except in the very unlikely event that there is a collision among the sampled points $\{(x^{(i)},y^{(i)})\}_i$, the distribution of values $f(x^{(i)},y^{(i)})$ will look uniformly random, regardless of whether $f$ is sampled uniformly from $F_\text{yes}$ or $F_\text{no}$---and so distinguishing is impossible.
The truth tables for $h$ are constructed from certain affine shifts of Reed--Muller codewords.

As for item (\textit{ii}), the passive quantum tester for this property first applies the unitary $U$ defined by $\ket{x,y,b}\mapsto\ket{x,y,b\oplus \langle x,y\rangle}$ to $\ket{f}$.
If $f$ is a MM function the result should be $h$, a function depending only on the first $n$ variables, while if $\ket{f}$ is far from MM functions, it will have noticeable dependence on coordinates $n+1,\ldots, 2n$.
This dependence can be measured by Fourier-sampling the transformed state.

\subsection{Comparing access models}

Quantum data is always at least as good as classical samples, and from \Cref{tab:q-vs-c} we see that for a growing list of properties, testing from quantum data is competitive with testing from classical \textit{queries}.
In fact, quantum data can be vastly more powerful than classical queries for testing.
An extremal example of this is the \textsc{Forrelation} problem, which can be tested from $\mc O(1)$ function state copies but requires $\Omega(2^{n/2})$ classical queries \cite{aaronson2015forrelation}.

Conversely, one may wonder to what extent classical queries may outperform quantum data for property testing.
An answer is not so obvious.
Although classical queries enable direct access to $f(x)$ at any point $x$ of the algorithm's choosing---a powerful advantage over quantum data---it is not so clear whether this can lead to a separation for property testing.
Recall that for a property testing problem, \textit{yes} and \textit{no} instances must be $\Omega(1)$-far in $L^1$ distance.
So to create a hard property for quantum data-based testers, one must find two sets of functions which pairwise differ on a \textit{constant fraction} of the locations in their truth tables, yet still remain hard to distinguish by a quantum algorithm operating on copies of their function states.

We succeed in ``hiding'' these large differences and identify a testing problem for which classical queries have a dramatic advantage over quantum data.

\begin{theorem}\label{inf-thm:separating-passive-quantum-from-query-classical-testing-new}
    There exists a testing task ($3$-fold intersection detection) that can be accomplished with $\mc O(1)$ classical queries but requires $\Omega(2^{n/2})$ copies for quantum testing from data.
\end{theorem}

\noindent Combined with the \textsc{Forrelation} separation of \cite{aaronson2015forrelation}, \Cref{inf-thm:separating-passive-quantum-from-query-classical-testing-new} entails that quantum data and classical queries are (essentially) maximally incomparable.
See \Cref{fig:different-access-models} for a full picture of resource inequalities for testing.

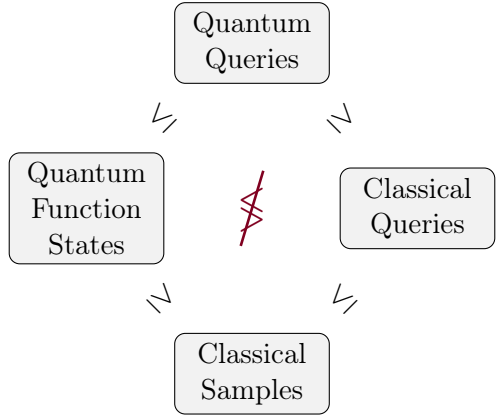
\begin{figure}
    \centering
    \begin{tikzpicture}[scale=1, transform shape] 
        \tikzstyle{box} = [rectangle, rounded corners, minimum width=2cm, minimum height=0.8cm, draw, fill=gray!10, text centered, font=\normalsize]
    
        \useasboundingbox (-3.2, -2.8) rectangle (3.2, 2.8);
        
        \node[box] (quantum_queries_top) at (0, 2.2) {\parbox{1.8cm}{\centering Quantum \\ Queries}};
        \node[box] (classical_queries_right) at (2.2, 0) {\parbox{1.8cm}{\centering Classical \\ Queries}};
        \node[box] (quantum_samples_left) at (-2.2, 0) {\parbox{1.8cm}{\centering Quantum \\ Function \\ States}};
        \node[box] (classical_samples_bottom) at (0, -2.2) {\parbox{1.8cm}{\centering Classical \\ Samples}};
    
        \node at (1.2, 1.2) [rotate=315] {\scalebox{1.0}[1.0]{\large $\geq$}};
        \node at (1.2, -1.2) [rotate=45] {\scalebox{1.0}[1.0]{\large $\leq$}};
        \node at (-1.2, -1.2) [rotate=-45] {\scalebox{1.0}[1.0]{\large $\geq$}};
        \node at (-1.2, 1.2) [rotate=45] {\scalebox{1.0}[1.0]{\large $\leq$}};
        
        \node at (0, 0) {\textcolor{burgundy}{\scalebox{1.0}[1.5]{\large $\lessgtr$}}}; 
        \draw[-, burgundy, line width=1pt] (-0.15, -0.5) -- (0.15, 0.5); 
    \end{tikzpicture}
    \caption{\textbf{Property testing resource inequalities.} The figure illustrates the connections between four different data access models in property testing, namely classical/quantum example/query access. 
    Here, ``resource A $\geq$ resource B'' means that access to resource B can be simulated from access to resource A without any overhead. (For example, a single classical query can be used to simulate a single classical sample.)
    As a consequence of \Cref{inf-thm:separating-passive-quantum-from-query-classical-testing-new} and \cite{bernstein1997complexity, simon1997power, aaronson2015forrelation}, the only two among these access models that are not trivially comparable are in fact very incomparable.
    }
    \label{fig:different-access-models}
\end{figure}

\Cref{inf-thm:separating-passive-quantum-from-query-classical-testing-new} is proved in \Cref{section:separating-classical-queries-from-function-states} as \Cref{thm:qsamp-cquery-sep}.
Given a function $f:\{0,1,2\}\times\{0,1\}^n\to\{0,1\}$ that indicates three subsets of the hypercube $A,B,C\subseteq\{0,1\}^n$, the 3-fold intersection detection task is to determine if the fractional 3-fold intersection $|A\cap B\cap C|/2^n$ is $0$ or $\Omega(1)$-far from $0$.

This property is readily tested from queries by computing the probability $x\in A\wedge x\in B\wedge x\in C$ for uniformly-random $x$.
To prove the quantum passive testing lower bound, we show the indistinguishability of two ensembles of function states encoding set triples
\[\big\{(A,B,C)\big\}_{A,B,C\;\overset{\mathsf{iid}}{\sim}\; \mc P\{0,1\}^n} \qquad \text{vs.}\qquad\big\{(A,B,A\Delta B)\}_{A,B\;\overset{\mathsf{iid}}{\sim}\; \mc P\{0,1\}^n}\, .\]
Here $\Delta$ denotes symmetric difference, $\mc P$ denotes the power set, and the samples are uniform.
Note the first ensemble has mutual intersection of $\Omega(1)$ density with high probability, while the ensemble always has zero intersection.
To obtain the lower bound, the main observation is that the $t$-copy versions of the two associated function state ensembles are \emph{equal} when projected onto the so-called \textit{distinct subspace} (\emph{i.e.,} the subspace spanned by basis states for which the $t$ input registers are distinct).
This projection moves the state ensembles at most $\mc O(t/2^{n/2})$ in trace distance, so we conclude that for any $t =o(2^{n/2})$, the two ensembles cannot be distinguished using $t$ function state copies.

While $3$-fold intersection detection is provably difficult from function state copies, its $2$-fold counterpart is quantumly easy. 
For example, we can perform SWAP tests between function state copies.
Rather surprisingly, we find that 2-fold intersection detection can even be achieved by a quantum algorithm using unentangled measurements applied to each copy of the quantum state independently. 
Namely, in \Cref{sec:two-fold-intersection}, we provide a quantum Fourier sampling-based approach to 2-fold intersection detection.

\subsection{A challenge: lower bounds for quantum monotonicity testing}

We also begin the project of finding lower bounds for the passive quantum testing model.
Our main contribution is to establish lower bounds for monotonicity as an important first open problem.
In particular, we show the ensembles that entail strong lower bounds for classical passive testing are wholly inadequate for quantum passive testing.

\begin{theorem}
    \label{thm:mon-lower-bound-difficulty}
    The ensembles in Goldreich et al. \cite{Goldreich2000} and Black \cite{black_nearly_2023} can be distinguished by a quantum algorithm with $\mc O(1/\epsilon)$ copies of the corresponding function states.
\end{theorem}

\noindent This theorem says that the best lower bound such ensembles could imply for quantum passive testing is $\Omega(1/\epsilon)$.
But that is no better than the lower bound that exists generically for every (non-trivial) property.\footnote{A classical query complexity lower bound of $\Omega(1/\varepsilon)$ also holds for testing any non-trivial property \cite{fischer2024basiclowerboundproperty}.}
To see that $\Omega(1/\epsilon)$ holds generically, it suffices to consider only two functions, $f_{\mathrm{yes}}$ and $f_{\mathrm{no}}$, that are exactly $\varepsilon$-far apart.
This is equivalent to $\braket{f_{\mathrm{yes}}|f_{\mathrm{no}}}= 1 - \varepsilon$, so the trace distance between $\ket{f_{\mathrm{yes}}}^{\otimes t}$ and $\ket{f_{\mathrm{no}}}^{\otimes t}$ is $\sqrt{1-(1-\varepsilon)^{2t}}\leq \sqrt{2t\varepsilon}$.
Therefore, distinguishing between $f_{\mathrm{yes}}$ and $f_{\mathrm{no}}$ with success probability $\geq 2/3$ requires $t\geq\Omega(1/\varepsilon)$ copies of the respective function state.
Note that this straightforward lower bound separates quantum samples from quantum queries in the regime $\epsilon=\mathcal{O}(n^{-3/2})$ because of the $\widetilde{\mathcal{O}}(n^{1/4}/\varepsilon^{1/2})$ quantum query upper bound for monotonicity testing proved in \cite{Belovs2015}.

We prove \Cref{thm:mon-lower-bound-difficulty} via a combinatorial analysis of the spectrum of the matrix \[A:=\E_{\psi\sim E_0}\psi^{\otimes t}-\E_{\phi\sim E_1}\phi^{\otimes t}\,,\] where $E_0$ and $E_1$ are the ``yes'' and ``no'' ensembles from \cite{Goldreich2000} (or, later, from \cite{black_nearly_2023}). 
As neither of our ensembles is close to Haar-random, we cannot directly draw on the rich recent literature on quantum pseudorandomness \cite{ji2018pseudorandom, brakerski2019pseudorandom, giurgicatiron2023pseudorandomnesssubsetstates, jeronimo2024pseudorandompseudoentangledstatessubset, metger2024simpleconstructionslineardepthtdesigns, chen2024efficientunitarydesignspseudorandom, schuster2024randomunitariesextremelylow, ma2024constructrandomunitaries}.
Instead, we notice that our function state is unitarily equivalent to a phase state for a closely-related Boolean function. 
An intricate index rearrangement reveals $A$ to be block-diagonal, with each block interpretable as the adjacency matrix for a complete bipartite graph.
We then determine the spectrum of each block, with eigenvalues and their multiplicities given as functions of certain combinatorial quantities.
Exponential generating function techniques lead to explicit formulas for these quantities, and finally the asymptotics can be understood by taking a probabilistic perspective on the counting formulas.
Concentration arguments finish the proof and allow us to conclude that $\|A\|_1\geq \Omega(1)$ (and thus the two ensembles are distinguishable) as soon as $t=\Omega(\epsilon^{-1})$.
This argument is presented in detail in \Cref{section:passive-quantum-monotonicity-testing-lower-bound}.

A final remark for this section: for certain regimes of $\varepsilon$, the $\Omega(1/\varepsilon)$ lower bound on the number of function state copies already separates passive quantum testing from quantum query-based testing. For example, (adaptive) quantum query complexity upper bounds of $\tilde{\mathcal{O}}(n^{1/4}/\varepsilon^{1/2})$ for monotonicity testing \cite{Belovs2015} and of $\tilde{\mathcal{O}}((k/\varepsilon)^{1/2})$ for $k$-junta testing \cite{ambainis_efficient_2015} are known. However, to the best of our knowledge, the ``correct'' $\varepsilon$-scaling for quantum property testing of classical functions is far from understood; prior works such as \cite{atici2007quantum, buhrman2008quantum, chakraborty2010newresults, aaronson2015forrelation, ambainis_efficient_2015, montanaro2016survey} seem to establish quantum query complexity lower bounds only for constant $\varepsilon$. 

\subsection{Outlook and future directions}

Our results, (most of them) summarized in \Cref{table:results-overview}, highlight passive quantum property testing as a rich testing model deserving of concerted study.
We grow the list of properties with efficient passive quantum testers, introduce new techniques for testing, show that the abilities of passive quantum testing extend beyond the reach of Fourier sampling, and highlight subtleties in comparing classical and quantum resources for property testing.

\begin{table}[htbp]
\begin{center}
\resizebox{\textwidth}{!}{%
\begin{tabular}{||c | c c c||} 
 \hline
  & Quantum function states & $\begin{matrix}
  \text{Classical}\\
  \text{samples}
  \end{matrix}$ & $\begin{matrix}
  \text{Classical}\\
  \text{queries}
  \end{matrix}$ \\[5pt] 
 \hline\hline{} & {}& \cellcolor{lightgray!15}{}&\cellcolor{lightgray!15}{}\\[-1em]
  $\begin{matrix}
      \text{Monotonicity}\\ 
      \text{testing}
  \end{matrix}$
  &  $\begin{matrix}
  \tilde{\mathcal{O}}(n^2 / \varepsilon^2) \\
  \text{\Cref{theorem:quantum-monotonicity-testing}}
  \end{matrix}$ & \cellcolor{lightgray!15}$\begin{matrix}
  \exp\left(\Omega\left(\min\{\sqrt{n}/\varepsilon,n\}\right)\right) \\
  \text{\cite{Goldreich2000, black_nearly_2023}}
  \end{matrix}$ &  \cellcolor{lightgray!15}$\begin{matrix}
  \tilde{\mathcal{O}}\left(\min\left\{n/\varepsilon, \sqrt{n}/\varepsilon^2\right\}\right) \\
  \text{\cite{Goldreich2000, DBLP:journals/siamcomp/KhotMS18}}
  \end{matrix}$\\ 
  \hline{} & {}& \cellcolor{lightgray!20}{}&\cellcolor{lightgray!15}{}\\[-1em]
 $\begin{matrix}
      \text{Symmetry}\\ 
      \text{testing}
  \end{matrix}$
  &   $\begin{matrix}
  \tilde{\mathcal{O}}(1/\varepsilon^2) \\
  \text{\Cref{theorem:quantum-symmetry-testing}}
  \end{matrix}$ & \cellcolor{lightgray!15}$\begin{matrix}
  \Theta (n^{1/4}) \\
  \text{\cite{alon2016active}}
  \end{matrix}$ &  \cellcolor{lightgray!15}$\begin{matrix}
  \mathcal{O}(1/\varepsilon)\\
  \text{\cite{blais2015partially}}
  \end{matrix}$\\ 
  \hline{} & {}& {} &\cellcolor{lightgray!15}{}\\[-1em]
 $\begin{matrix}
      \text{Triangle-freeness}\\ 
      \text{testing}
  \end{matrix}$
  &  $\begin{matrix}
  \tilde{\mathcal{O}}\left(\left(\mathsf{Tower}\left(C\cdot \left\lceil \log\left(\frac{1}{\varepsilon}\right)\right\rceil\right)\right)^{6}\right) \\
  \text{\Cref{theorem:quantum-triangle-freeness-testing}}
  \end{matrix}$ & $\begin{matrix}
  \Omega(2^{n/3}) \\
  \text{\Cref{remark:triangle-freeness-lower-bound}}
  \end{matrix}$ & \cellcolor{lightgray!15}$\begin{matrix}
  \mathcal{O}\left(\mathsf{Tower}\left(C\cdot \left\lceil \log\left(\frac{1}{\varepsilon}\right)\right\rceil\right)\right) \\
  \text{via \cite{fox2011new, hatami2016arithmetic}}
  \end{matrix}$\\ 
  \hline{} & {}& {}&{}\\[-1em]
 $\begin{matrix}
      \text{3-fold intersection}\\ 
      \text{estimation}
  \end{matrix}$
  &  $\begin{matrix}
  \Omega\left( 2^{n/2} \right) \\
  \text{\Cref{inf-thm:separating-passive-quantum-from-query-classical-testing-new}}
  \end{matrix}$ & $\begin{matrix}
  \Omega\left( 2^{n/2} \right) \\
  \text{via \Cref{inf-thm:separating-passive-quantum-from-query-classical-testing-new}}
  \end{matrix}$ &  $\begin{matrix}
  \mathcal{O}(1) \\
  \text{\Cref{inf-thm:separating-passive-quantum-from-query-classical-testing-new}}
  \end{matrix}$\\
  \hline{} & \cellcolor{lightgray!15}{}& \cellcolor{lightgray!15}{}&\cellcolor{lightgray!15}{}\\[-1em]
  
 $\begin{matrix}
      \textsc{Forrelation}
  \end{matrix}$ & \cellcolor{lightgray!15}$\begin{matrix}
      \mathcal{O}(1) \\ \text{\cite{aaronson2015forrelation}}
  \end{matrix}$ & \cellcolor{lightgray!15}$\begin{matrix}
      \tilde{\Omega}\left(2^{n/2}\right) \\ \text{\cite{aaronson2015forrelation}}
  \end{matrix}$ & \cellcolor{lightgray!15}$\begin{matrix}
      \tilde{\Theta}\left(2^{n/2}\right) \\ \text{\cite{aaronson2015forrelation}}
  \end{matrix}$ \\
  \hline
\end{tabular}
} 
\caption{ 
\textbf{Our bounds in context.} 
The table contrasts our results on property testing from quantum function states with results from the literature (in gray).
Where the $\varepsilon$-dependence is not shown explicitly, we have set $\varepsilon$ to some suitably small positive constant value.
For monotonicity, symmetry, and triangle-freeness, passive quantum testing from function states is (at least) exponentially easier than passive classical testing from samples and at most polynomially harder than classical testing from queries.
The testing problem derived from 3-fold intersection estimation is complementary to the \textsc{Forrelation} problem in that quantum function/phase states and classical queries swap roles in the exponential separation.
}\label{table:results-overview}
\end{center}
\end{table}

In fact, it seems passive quantum testing can make good on the promise of testing from data where classical passive testing cannot.
With passive quantum testing, it is possible to generate a dataset about a Boolean function without foreknowledge of the property one would eventually like to test, and still be assured (for a growing list of properties) that testing will be very efficient.
In particular, these results suggest that \textit{quantum data}, rather than classical data, could enable the application to machine learning imagined in \cite{goldreich1998property}: as an inexpensive preprocessing procedure that informs the choice of suitable, more data-intensive learning algorithms.

Here we lay out some directions for future work.

\paragraph{More and improved bounds for passive quantum property testing.}
We have established upper bounds for passive quantum testing of monotonicity, symmetry, and triangle-freeness from function states.
These three properties together with linearity testing \cite{bernstein1997complexity} and junta testing \cite{atici2007quantum, ambainis_efficient_2015} already demonstrate the power of quantum data for testing a variety of quite different properties, and it seems important to explore quantum datasets in the context of other testing problems.
As highlighted in \Cref{tab:q-vs-c}, quantum low-degree testing of Boolean functions is a natural next challenge, with the more general class of locally characterized affine-invariant properties \cite{bhattacharyya2013every} as a longer-term goal.

One may aim to tighten our bounds to precisely pin down the power of quantum data for these testing tasks.
Here, having established that the constructions from classical passive monotonicity testing lower bounds are inadequate for the quantum case, we consider it especially interesting to obtain a $n$-dependent lower bound for passive quantum monotonicity testing in the constant $\epsilon$ regime. 
Settling the $n$-dependence of the quantum sample complexity for passive quantum monotonicity testing is a tantalizing question for future work.

\paragraph{Characterizing properties with constant-complexity passive quantum testers}

In the classical case, \cite{blaisyoshida2019} gave a complete characterization of those properties that can be tested with a constant number of samples.
Achieving an analogous characterization for properties that can be tested from constantly-many function state copies would help demarcate the boundary of quantum advantage for this model. 

Intriguingly, the quantum case raises a further question about the constant complexity regime: For properties that admit a constant-copy passive quantum testers, can this always be achieved by algorithms that do not use entangled multi-copy measurements?
The role of single- versus multi-copy quantum processing has recently been explored in the literature on learning and testing for quantum objects (see, e.g., \cite{chen2022exponential, huang2022quantum-advantage, caro2023learning,  hinsche2024singlecopystabilizertesting}) and in quantum computational learning theory \cite{arunachalam2023optimal}, but the picture is far from clear for properties of function states (and of pure states more generally).
Concretely, while our testers for monotonicity and symmetry are single-copy algorithms, our triangle-freeness tester uses two-copy SWAP tests and there does not seem to be an immediate way of replacing this by single-copy quantum processing.

One may also ask about the necessity of auxiliary quantum systems in quantum sample-based testers with constant sample complexity; for example, our symmetry tester relied on auxiliary systems to implement the symmetric subspace projector.
The number of available auxiliary systems is already known to play an important role in, for instance, Pauli channel learning \cite{chen2022quantum, chen2024tight, chen2024efficientpaulichannelestimation}, and exploring its relevance for constant-complexity passive quantum testing may shed new light on how these quantum testers achieve their better-than-classical performance.

\paragraph{Other quantum datasets for classical properties}
We have considered only one kind of quantum representation of classical functions: coherent superpositions of evaluations of $f$ (as function states).
Already these are enough to gain major advantages over testing from classical data, but one could ask for more.
Are there other, better quantum datasets that lead to even faster testers or extend quantum advantage to more properties?
To keep this question interesting, one would require that the dataset be not too tailored to any property.

In fact, this question may be best phrased as a sort of ``compression game'': we are first given a very long list of questions that we might be asked regarding some black-box function $f$.
We then have $T(n)$ time to interact with an oracle for $f$, during which period we generate whatever data we would like.
What is the best quantum dataset to generate, so that we are best prepared to answer a random (or perhaps worst-case) question from the list?

\paragraph{Passive quantum testing for quantum properties.}
Recently there has been a growing interest in property testing for quantum objects, such as states \cite{harrow2013testing, odonnell2015quantum, harrow2017sequential, carmeli2017probing, buadescu2020lower, gross2021testing, soleimanifar2022testing, chen2022exponential, grewal2023improved, arunachalam2024toleranttestingstabilizerstates, chen2024stabilizerbootstrappingrecipeefficient, arunachalam2024notepolynomialtimetoleranttesting, bao2024toleranttestingstabilizerstates}, unitaries \cite{dallagnol2022quantumproofsof, chen2023testing, she2022unitary}, channels \cite{chen2022exponential, aharonov2022quantum, bao2023testing}, and Hamiltonians \cite{laborde2022quantum, bluhm2024hamiltonian, arunachalam2024testinglearningstructuredquantum}.
It is an interesting challenge to design datasets to enable passive versions of these tasks.
Just as in the above, we would want quantum datasets that are mostly agnostic to the property to be tested.

In fact, some existing work can be viewed as advocating for quantum datasets.
When restricting ourselves to collecting classical data, classical shadows \cite{huang2020predicting, huang2022learning} serve as a useful representation, but place restrictions on the properties that can be tested after-the-fact. 
Shadow tomography procedures \cite{aaronson2018shadow, badescu2021improved, caro2023learning} can remove such restrictions but use multi-copy measurements that depend on the properties of interest, and thus in general seem to require quantum data storage to enable passivity.\footnote{(Non-adaptive) Pauli shadow tomography \cite{huang2021information,caro2023learning,king2024triplyefficientshadowtomography,chen2024optimaltradeoffsestimatingpauli} in some sense interpolates between the (dis-)advantages of classical shadows and shadow tomography for the current discussion: When promised in advance that the properties in question are characterized by expectation values of arbitrary Pauli observables, some of the relevant data can be collected and stored classically in advance, without knowing which specific Pauli observables matter. However, part of the quantum processing still requires knowing the specific Paulis of interest, so to achieve passivity, it seems that some data still has to be stored quantumly.}
The relevance of data storage in a quantum memory for certain quantum process learning tasks has also been explored in \cite{bisio2010optimal, bisio2011quantum, sedlak2019optimal, lewandowska2022storage, lewandowska2024storage}.
In this context, our work can be viewed as investigating the power of quantum data, stored in quantum memory, for testing properties of diagonal unitary processes arising from classical Boolean functions. 
We hope that this will inspire future attempts at using quantum data as a resource for passively quantumly testing properties of more general quantum processes.

\subsection{Related work}

\paragraph{Passive classical property testing.\protect\footnote{
Due to the vastness of the area of property testing, even when restricting the focus to passive testing, this paragraph is intended to provide context for our work rather than an exhaustive bibliography for the field.}}
Passive (or sample-based) property testing goes back to \cite{goldreich1998property} (see also \cite{kearns1998testing}), who introduced it as a testing counterpart to Valiant's model of probably approximately correct (PAC) learning 
\cite{valiant1984theory}. In particular, \cite[Proposition 3.1.1]{goldreich1998property} observes that PAC learners give rise to passive testers (see also \cite[Proposition 2.1]{ron2008property}).
Later, \cite{balcan_active_2012} proposed active testing as a model interpolating between sample- and query-based testing.
Both for passive and active testing, and for a variety of problems, several works have established lower bounds separating them from the more standard query-based testing model. 
Some notable examples of tasks with such separations include ($k$-)linearity \cite{balcan_active_2012, alon2016active}, $k$-juntaness \cite{alon2016active}, (partial) symmetry \cite{blais2015partially, alon2016active}, low-degreeness \cite{alon2016active}, and monotonicity \cite{Goldreich2000, black_nearly_2023}. 
We present these results and how they compare to quantum testing in \Cref{tab:q-vs-c}.

\cite{blaisyoshida2019} gave a full characterization of properties of Boolean-valued functions that admit passive testing with a constant (\textit{i.e.}, independent of domain size) number of uniformly random samples, demonstrating that this is indeed a relatively restricted type of properties.
While the works mentioned so far have focused on the case of uniformly random data points (or, in the case of active learning, uniformly random sets of admissible query points), more recently there has been renewed interest in passive distribution-free testing, see for instance \cite{halevy2007distribution-free, blais2021vc-dimension}.
Finally, the framework of passive testing has also been explored for objects other than Boolean functions, especially for testing geometric properties \cite{matulef2010testing, balcan_active_2012, kothari2014testing, neeman2014testing, chen2017sample-based, berman2019power, berman2019testing}.

\paragraph{Quantum property testing.}
In our work, the focus is on quantumly testing properties of classical functions. This topic, considered for example in \cite{atici2007quantum, chakraborty2010newresults, hillery2011quantum, aaronson2015forrelation, ambainis_efficient_2015}, is one of the main directions in quantum property testing, an area that goes back to \cite{buhrman2008quantum} and is surveyed in \cite{montanaro2016survey}.
However, quantum property testing also considers quantum algorithms that test properties of other classical objects from quantum data access. 
Notable examples of other objects to quantumly test include probability distributions \cite{bravyi2011quantum, chakraborty2010newresults, gilyen2020distributional}, graphs \cite{ambainis2011quantum, chakraborty2010newresults}, and groups \cite{friedl2009quantum, inui2011quantum}.
Finally, recently there have also been significant insights in quantum property testing for quantum objects, notably states \cite{harrow2013testing, odonnell2015quantum, harrow2017sequential, carmeli2017probing, buadescu2020lower, gross2021testing, soleimanifar2022testing, chen2022exponential, grewal2023improved, hinsche2024singlecopystabilizertesting, arunachalam2024toleranttestingstabilizerstates, chen2024stabilizerbootstrappingrecipeefficient, arunachalam2024notepolynomialtimetoleranttesting, bao2024toleranttestingstabilizerstates, mehraban2024improvedboundstestinglow}, unitaries \cite{dallagnol2022quantumproofsof, chen2023testing, she2022unitary}, channels \cite{chen2022exponential, aharonov2022quantum, bao2023testing}, and Hamiltonians \cite{laborde2022quantum, bluhm2024hamiltonian, arunachalam2024testinglearningstructuredquantum}.

\section{Passive quantum testing upper bounds}\label{section:passive-quantum-testing-upper-bounds}

\subsection{Defining passive quantum property testing}

As outlined in \Cref{section:introduction}, passive property testing considers testing from (non-adaptively chosen) data that does not depend on the property to be tested.
We propose a quantum version of this model by considering quantum testing algorithms that have access to copies of a quantum data state.
Here, we consider the following form of quantum data encoding for a Boolean function $f:\{0,1\}^n\to\{0,1\}$: We work with \emph{function states}
\begin{equation}\label{eq:function-state}
    \ket{f}
    =\ket{\Psi_f}=\frac{1}{\sqrt{2^n}}\sum_{x\in\{0,1\}^n} \ket{x,f(x)}\, .
\end{equation}
When the function $f$ is clear from context, we will also use the notation $\ket{\Psi}=\ket{\Psi_f}$. Natural variations of this notation, \textit{e.g.}, $\ket{\Psi'}=\ket{\Psi_{f'}}$, will also be used.

With this, we can now formally define the notion of passive quantum property testing for Boolean functions.

\begin{definition}[Passive quantum property testing]\label{definition:passive-quantum-testing}
    Let $\mathcal{P}_n\subseteq\{0,1\}^{\{0,1\}^n}$ be some property of Boolean functions on $n$ bits, let $\delta, \varepsilon\in (0,1)$.
    A quantum algorithm is a \emph{passive quantum tester} with accuracy/distance parameter $\varepsilon$ and confidence parameter $\delta$ for $\mathcal{P}_n$ from $m=m(\varepsilon,\delta)$ function state copies if the following holds:
    When given $m$ copies of $\ket{\Psi_f}$, the quantum algorithm correctly decides, with success probability $\geq 1-\delta$, whether
    \begin{itemize}
        \item[(i)] $f\in\mathcal{P}_n$, or
        \item[(ii)] $\Pr_{x\sim\{0,1\}^n}[f(x)\neq g(x)]\geq \varepsilon$ holds for all $g\in \mathcal{P}_n$,
    \end{itemize}
    promised that $f$ satisfies either (i) or (ii).
\end{definition}

This work explores \Cref{definition:passive-quantum-testing} in the context of several properties. 
We develop testers for monotonicity, symmetry, and triangle-freeness. 
The main ideas underlying our algorithms are introduced in \Cref{subsection;passive-quantum-monotonicity-testing}.
In the next section, we present the passive quantum monotonicity tester in detail, while the symmetry and triangle-freeness testers are deferred to \Cref{sec:passive-quantum-testers}. 

\subsection{Passive quantum monotonicity testing}\label{subsection;passive-quantum-monotonicity-testing}

We define the natural partial order $\preceq$ on the Boolean hypercube $\{0,1\}^n$ via $x \preceq y~:\Leftrightarrow ~ (x_i\leq y_i$ holds for all $1\leq i\leq n)$.
A function $f:\{0,1\}^n\to\{0,1\}$ is called monotone if $f(x)\leq f(y)$ holds for all $x,y\in\{0,1\}^n$ with $x \preceq y$. The associated classical testing problem can be formulated as follows:

\begin{problem}[Classical monotonicity testing]\label{problem:classical-monotonicty-testing}
  Given query access to an unknown function $f:\{0,1\}^n\to\{0,1\}$ and an accuracy parameter $\varepsilon\in (0,1)$, decide with success probability $\geq 2/3$ whether
  \begin{itemize}
      \item[(i)] $f$ is monotone, or
      \item[(ii)] $f$ is $\varepsilon$-far from all monotone functions, that is, we have $\Pr_{x\sim\{0,1\}^n}[f(x)\neq g(x)]\geq\varepsilon$ for all monotone functions $g:\{0,1\}^n\to\{0,1\}$,
  \end{itemize}  
  promised that $f$ satisfies either (i) or (ii).  
\end{problem}

Here, as well as in our other property testing tasks below, will think of (i) as the accept case and of (ii) as the reject case. This then allows us to speak of completeness (for getting case (i) right) and soundness (for getting case (ii) right).
Here, the chosen success probability of $2/3$ is an arbitrary constant $>1/2$, it can be boosted arbitrarily close to $1$ through repetition and majority voting.

As introduced above, monotonicity is a global property of a function. However, there is a straightforward equivalent local formulation:

\begin{proposition}[Local characterization of monotonicity]\label{proposition:local-characterization-monotonicity}
    A function $f:\{0,1\}^n\to\{0,1\}$ is monotone if and only if for all $x\in\{0,1\}^n$ and for all $i\in [n]$, the following holds:
    \begin{equation}\label{eq:local-monotonicity-condition}
        \neg\left(( x_i=0 \land f(x)=1 \land f(x+e_i)=0)\lor (x_i=1 \land f(x)=0 \land f(x+e_i)=1)\right) \, ,
    \end{equation}
    where $e_i$ denotes the $i^{\mathrm{th}}$ standard basis vector.
\end{proposition}

It turns out that functions far from the set of all monotone functions necessarily violate \Cref{eq:local-monotonicity-condition} on a non-negligible fraction of all possible $x$ and $i$. This makes it possible to test for monotonicity by checking \Cref{eq:local-monotonicity-condition} on a small number of randomly chosen $x$ and $i$.

\begin{theorem}[Soundness of monotonicity testing (compare \cite{Goldreich2000})]\label{theorem:classical-monotonicity-testing-soundness}
    If $f:\{0,1\}^n\to\{0,1\}$ is exactly $\varepsilon$-far from all monotone functions, then 
    \small
    \begin{equation}\label{eq:classical-monotonicity-testing-soundness}
         \frac{\varepsilon}{n}
         \leq \Pr_{x\sim\{0,1\}^n, i\sim [n]}\left[( x_i=0 \land f(x)=1 \land f(x+e_i)=0)\lor (x_i=1 \land f(x)=0 \land f(x+e_i)=1)\right]
        \leq 2\varepsilon\, .
    \end{equation}
    \normalsize
\end{theorem}

Therefore, we can solve \Cref{problem:classical-monotonicty-testing} from only $\mathcal{O}(n/\varepsilon)$ many queries to the unknown function, which was exactly the celebrated conclusion of \cite{Goldreich2000}. While this query complexity does depend on $n$, the dependence is only logarithmic in the size of the function domain, and it in particular is exponentially better than the $n$-dependence in the query complexity of learning monotone functions \cite{DBLP:conf/focs/BlumBL98}.

Our passive quantum monotonicity tester also crucially relies on \Cref{theorem:classical-monotonicity-testing-soundness}.
Here, we first reinterpret the probability appearing in \Cref{eq:classical-monotonicity-testing-soundness} in terms of Fourier-analytic quantities, which we then estimate based on quantum Fourier sampling. 
Our procedure is summarized in \Cref{algorithm-monotonicity-tester-fourier}, and our next theorem establishes that it is both complete and sound.

\begin{algorithm}
  \caption{Monotonicity testing from quantum examples}
  \begin{algorithmic}[1]
        \item[\textbf{Input:}] accuracy parameter $\varepsilon\in (0,1)$; confidence parameter $\delta\in (0,1)$; $\tilde{\mathcal{O}}\left(\frac{n^{2} \log(1/\delta)}{\varepsilon^{2}}\right)$ many copies of a function state $\ket{f} = \frac{1}{\sqrt{2^n}}\sum_{x \in \{0, 1\}^n} \ket{x, f(x)}$. 
        \item[\textbf{Output:}] ``accept'' or ``reject''.
        \item[Initialization:] $\varepsilon_2=\frac{\varepsilon}{3}$, $\varepsilon_5=\frac{\varepsilon}{3n}$, $\delta_1=\delta_2=\frac{\delta}{3}$, $\delta_5=\frac{\delta}{3n}$, $m_1=\max\{3m_2, \left\lceil 18\ln(2/\delta_1)\right\rceil\}$, $m_2=\left\lceil \frac{n^2 \ln(2/\delta_2)}{2\varepsilon_2^2} \right\rceil$, $m_4=m_5=\left\lceil \frac{4\ln (2/\delta_5)}{\varepsilon_5^2}\right\rceil$ 
        \State Use $m_1$ many copies of $\ket{f}$ to produce $m_2$ many Fourier samples $S_1,\ldots, S_{m_2}\subseteq [n]$ from $g=(-1)^{f}$.
        \State Take $\hat{\mathbf{I}} = \frac{1}{m_2}\sum_{\ell=1}^{m_2} |S_\ell|$.
        \State Use $m_4$ many copies of $\ket{f}$ to generate $m_5$ many classical samples $(x_1,f(x_1)),\ldots,(x_{m_5}, f(x_{m_5})$ from $f$. 
        \For{$1\leq i\leq n$}
            \State Take $\tilde{g}_i = \frac{1}{m_5}\sum_{k=1}^{m_5} (-1)^{x_k\cdot e_i + f(x_k)}$.
        \EndFor
        \State Set $\hat{p} = \frac{1}{2n}\hat{\mathbf{I}} - \frac{1}{2n}\sum_{i=1}^n \tilde{g}_i$.
        \State If $\hat{p}\leq \varepsilon/3n$, conclude that $f$ is monotone and accept. If $\hat{p}\geq 2\varepsilon/3n$, conclude that $f$ is $\varepsilon$-far from all monotone functions and reject. 
    \end{algorithmic}\label{algorithm-monotonicity-tester-fourier}
\end{algorithm}

\begin{theorem}[Passive quantum monotonicity testing]\label{theorem:quantum-monotonicity-testing}
    \Cref{algorithm-monotonicity-tester-fourier} is an efficient quantum algorithm that uses $\tilde{\mathcal{O}}\left(\frac{n^{2} \log(1/\delta)}{\varepsilon^{2}}\right)$ copies of $\ket{f} = \frac{1}{\sqrt{2^n}}\sum_{x\in\{0,1\}^n} \ket{x, f(x)}$ to decide, with success probability $\geq 1-\delta$, whether $f$ is monotone or $\varepsilon$-far from monotone.
\end{theorem}

In particular, \Cref{theorem:quantum-monotonicity-testing} shows that passive quantum testers can exponentially outperform the classical passive monotonicity 
testing lower bound of $\exp\left(\Omega\left(\min\{\sqrt{n}/\varepsilon,n\}\right)\right)$ \cite{Goldreich2000, black_nearly_2023}.

\begin{proof}

    We begin with a useful rewriting of the probability from \Cref{eq:classical-monotonicity-testing-soundness}.
    To this end, as is commonly done in the analysis of Boolean functions, consider the induced function $g:\{-1,1\}^n\to\{-1,1\}$ obtained from $f$ via the relabeling $0\leftrightarrow 1$ and $1\leftrightarrow -1$.
    Next, we recall the definition of the $i$th derivative in Boolean analysis (compare, \textit{e.g.}, \cite[Definition 2.16]{odonnell2021analysisbooleanfunctions}): For $i\in [n]$, 
    \begin{align*}
        D_i g(x)
        \coloneqq \frac{g(x^{(i\mapsto 1)}) - g(x^{(i\mapsto -1)})}{2}\, ,
    \end{align*}
    where we used the notation $x^{(i\mapsto b)}$ to denote the $n$-bit string obtained from $x$ by replacing the $i$th bit with $b$.
    Consequently, we can compute
    \begin{align*}
        \frac{(D_i g(x))^2- D_i g(x)}{2}
        &=\begin{cases}
            1\quad &\textrm{if } g(x^{(i\mapsto 1)}) = -1 \land g(x^{(i\mapsto -1)}) = 1\\
            0 &\textrm{otherwise}
        \end{cases}\\
        &=\begin{cases}
            1\quad &\textrm{if } f(x^{(i\mapsto 0)}) = 1 \land f(x^{(i\mapsto 1)}) = 0\\
            0 &\textrm{otherwise}
        \end{cases}\, .
    \end{align*}
    Therefore, we can now rewrite our probability of interest as
    \begin{align*}
        &\Pr_{x\sim\{0,1\}^n, i\sim [n]}\left[( x_i=0 \land f(x)=1 \land f(x+e_i)=0)\lor (x_i=1 \land f(x)=0 \land f(x+e_i)=1)\right]\\
        &= \mathbb{E}_{x\sim\{0,1\}^n, i\sim [n]}\left[\frac{(D_i g(x))^2- D_i g(x)}{2}\right]\\
        &= \frac{1}{2}\mathbb{E}_{i\sim [n]}\mathbb{E}_{x\sim\{0,1\}^n}\left[(D_i g(x))^2\right] - \frac{1}{2}\mathbb{E}_{i\sim [n]}\mathbb{E}_{x\sim\{0,1\}^n}\left[D_i g(x)\right]\\
        &= \frac{1}{2}\mathbb{E}_{i\sim [n]}\left[\mathbf{Inf}_i [g]\right] - \frac{1}{2n}\sum_{i=1}^n \widehat{g}(\{i\})\\
        &= \frac{1}{2n}\mathbf{I}[g] - \frac{1}{2n}\sum_{i=1}^n \widehat{g}(\{i\})\, ,
    \end{align*}
    where the second-to-last step used the definition of the $i$th influence (compare \cite[Definition 2.17]{odonnell2021analysisbooleanfunctions}) as well as \cite[Proposition 2.19]{odonnell2021analysisbooleanfunctions}, and where the last step used the definition of the total influence (compare \cite[Definition 2.27]{odonnell2021analysisbooleanfunctions}).

    With this rewriting established, let us first analyze the probabilities that the different steps of \Cref{algorithm-monotonicity-tester-fourier} succeed and discuss what this implies for the estimator $\hat{p}$.
    Then, we will see how this gives rise to completeness and soundness.
    We have the following:
    \begin{itemize}
        \item Using the procedure of \cite{bernstein1997complexity}, one copy of $\ket{f}$ suffices to produce one Fourier sample from $g=(-1)^f$ -- that is, an $n$-bit string sampled from the probability distribution $\{|\hat{g}(S)|^2\}_{S\subseteq [n]}$ -- with success probability $1/2$. Additionally, one knows whether the sampling attempt was successful.\footnote{To see this, note that the procedure works as follows: Apply $H^{\otimes (n+1)}$; measure the last qubit; abort if that produces a $0$, continue if produces a $1$; measure the first $n$ qubits to produce an $n$-bit string. \label{footnote:quantum_fourier_sampling}} So, by simply repeating the above $m_1$ many times, we see that Step 1 succeeds in producing $m_2$ Fourier samples with success probability $\geq 1-\delta_1$.
        \item By a standard Chernoff-Hoeffding bound, we have $\lvert\hat{\mathbf{I}} - \mathbb{E}_{S\sim \mathcal{S}_g}[|S|]\rvert\leq\varepsilon_2$ with success probability $\geq 1-\delta_2$. Here, $\mathcal{S}_g$ denotes the Fourier distribution of $g$, defined via $\mathcal{S}_g(S)=|\widehat{g}(S)|^2$. 
        \item For any $1\leq i\leq n$, by a standard Chernoff-Hoeffding bound, Step 5 produces a $\varepsilon_5$-accurate estimate $\tilde{g}_i$ of $\widehat{g}(\{i\})$ with probability $\geq 1 - \delta_5$.
    \end{itemize}
    Therefore, by a union bound, we have that, with overall  success probability $\geq 1-(\delta_1+\delta_2+n\delta_5) = 1-\delta$, the estimates $\hat{\mathbf{I}}$ and $\tilde{g}_i$ simultaneously satisfy $\lvert\hat{\mathbf{I}} - \mathbb{E}_{S\sim \mathcal{S}_g}[|S|]\rvert\leq \varepsilon_2$ and $|\tilde{g}_i - \widehat{g}(\{i\})|\leq \varepsilon_5$ for all $1\leq i\leq n$.
    We condition on this high probability success event for the rest of the proof.
    In this event, our rewriting of the probability of interest from the beginning of the proof implies:   
    \small
    \begin{align*}
        &\left\lvert \Pr_{x\sim\{0,1\}^n, i\sim [n]}\left[( x_i=0 \land f(x)=1 \land f(x+e_i)=0)\lor (x_i=1 \land f(x)=0 \land f(x+e_i)=1)\right] - \hat{p}\right\rvert\\
        &\leq \frac{1}{2n}\left\lvert \mathbf{I}[g] -\hat{\mathbf{I}} \right\rvert + \frac{1}{2n}\sum_{i=1}^n \left\lvert \widehat{g}(\{i\}) - \tilde{g}_i\right\rvert\\
        &= \frac{1}{2n}\left\lvert \mathbb{E}_{S\sim \mathcal{S}_g}[|S|] - \hat{\mathbf{I}}\right\rvert + \frac{1}{2n}\sum_{i=1}^n \left\lvert \widehat{g}(\{i\}) - \tilde{g}_i\right\rvert\\
        &\leq \frac{\varepsilon_2}{2n} + \frac{\varepsilon_5}{2}\\
        &\leq \frac{\varepsilon}{3n}\, ,
    \end{align*}
    \normalsize
    where we used the identity $\mathbf{I}[g]=\mathbb{E}_{S\sim \mathcal{S}_g}[|S|]$ (compare \cite[Theorem 2.38]{odonnell2021analysisbooleanfunctions}). 
    So, we see that $\hat{p}$ is a $(\varepsilon/3n)$-accurate estimate for our probability of interest.

    To prove completeness of \Cref{algorithm-monotonicity-tester-fourier}, assume $f$ to be monotone.
    Then, \Cref{proposition:local-characterization-monotonicity} and \Cref{theorem:classical-monotonicity-testing-soundness} together with the above inequality imply that $\hat{p}\leq \varepsilon/3n\leq \varepsilon/2n$, and thus the final step in \Cref{algorithm-monotonicity-tester-fourier} correctly concludes that $f$ is monotone and accepts.

    To prove soundness, assume $f$ to be $\varepsilon$-far from monotone. Then, the lower bound in \Cref{theorem:classical-monotonicity-testing-soundness} together with the above inequality implies that $\hat{p} \geq 2\varepsilon/3n$, and thus the final step in \Cref{algorithm-monotonicity-tester-fourier} correctly concludes that $f$ is  $\varepsilon$-far from monotone and rejects.

    The overall number of copies of $\ket{f}$ used by the algorithm is $m_1+m_4$. 
    Plugging in the choices of the different $m_i$, we see that
    \begin{align*}
        m_1+m_4
        &\leq \max\{3m_2, \left\lceil 18\ln(2/\delta_1)\right\rceil\} + \left\lceil \frac{4\ln (2/\delta_5)}{\varepsilon_5^2}\right\rceil\\
        &\leq \max\left\{3\left\lceil \frac{n^2 \ln(2/\delta_2)}{2\varepsilon_2^2} \right\rceil, \left\lceil 18\ln(2/\delta_1)\right\rceil\right\} + \left\lceil \frac{36 n^2 \ln (6n/\delta)}{\varepsilon^2}\right\rceil\\
        &\leq \max\left\{3\left\lceil \frac{9 n^2 \ln(6/\delta)}{2\varepsilon^2} \right\rceil, \left\lceil 18\ln(6/\delta)\right\rceil\right\} + \left\lceil \frac{36 n^2 \ln (6n/\delta)}{\varepsilon^2}\right\rceil \\
        &\leq \tilde{\mathcal{O}}\left(\frac{n^2 \log (1/\delta)}{\varepsilon^2}\right)\, ,
    \end{align*}
    where the $\tilde{\mathcal{O}}$ hides a logarithmic dependence on $n$. 

    The quantum computational efficiency of \Cref{algorithm-monotonicity-tester-fourier} follows immediately from the efficiency of quantum Fourier sampling. 
    The classical computational efficiency is immediately apparent from our sample complexity bounds and the fact that the classical computation is dominated by the complexity of computing the empirical averages in Steps 2 and 4.
\end{proof}

We further note that because of the second inequality in \Cref{theorem:classical-monotonicity-testing-soundness}, the above procedure and proofs can be modified to quantumly efficiently solve the \emph{tolerant} version (as defined in \cite{parnas2006tolerant}) of the monotonicity testing problem---\textit{i.e.}, distinguishing between $f$ being $\varepsilon_1$-close or $\varepsilon_2$-far from monotone---using $\tilde{\mathcal{O}}\left(\frac{n^{2} \log(1/\delta)}{(\varepsilon_2 - \varepsilon_1)^{2}}\right)$ copies of a quantum function state, assuming that $\varepsilon_2 > C n\varepsilon_1$ holds with $C>1$ some constant.\footnote{In more generality, one can see: If an inequality like \Cref{eq:classical-monotonicity-testing-soundness} holds with lower bound $\ell_n(\varepsilon)$ and upper bound $u_n(\varepsilon)$, satisfying $\ell_n(0)=0=u_n(0)$, then estimating the relevant probability to accuracy $\sim \ell_n(\varepsilon_2-\varepsilon_1)$ suffices for tolerant property testing in the parameter range where there is a constant $c\in (0,1/2)$ such that $\ell_n(\varepsilon_2) - c\cdot \ell_n(\varepsilon_2-\varepsilon_1) > u_n(\varepsilon_1)  + c\cdot \ell_n(\varepsilon_2-\varepsilon_1)$.}
Because of this restrictive assumption on how $\varepsilon_1$ and $\varepsilon_2$ relate, this still falls short of a general tolerant passive quantum monotonicity tester.

Let us also note room for a qualitative improvement in our passive quantum monotonicity tester. Classical query-based testers typically enjoy perfect completeness, \textit{i.e.}, they accept monotone functions with unit probability. In contrast, our tester can be made to accept monotone functions with probability arbitrarily close but not equal to $1$. We leave as an open question whether our passive quantum monotonicity tester can be modified to achieve perfect completeness, while enjoying similar guarantees on the quantum sample and time complexity of the procedure.

\section{Fourier sampling does not suffice}
\label{section:FS-not-enough-for-mm}

The passive quantum testers for symmetry and triangle-freeness given in \Cref{subsection;passive-quantum-symmetry-testing,subsection;passive-quantum-trianglefreeness-testing} notably do not use quantum Fourier sampling.
One might ask if this is really necessary, given that Fourier sampling (sometimes augmented by classical samples) suffices for so many other learning and testing tasks.
This section presents a class of functions, Maiorana--McFarland bent functions, which can be tested with $\mc O(1)$ function state copies, but any algorithm relying solely on Fourier samples and classical samples requires super-polynomial classical samples to succeed.

The \textit{Maiorana--McFarland} functions \cite[Section 6.1]{carlet2016four} on $2n$ bits, denoted $\mathrm{MM}_n$, are given by
\begin{align*}
f_h:\F_2^n\times \F_2^n &\to\F_2\\
(x,y) &\mapsto \langle x,y\rangle + h(x)
\end{align*}
where $h$ ranges over all functions $\F_2^n\to\F_2$.
Maiorana--McFarland functions are a subset of the class of \textit{bent} Boolean functions $g:\{0,1\}^m\to\{0,1\}$, which are those with $\widehat{g}(S)^2=1/2^m$ for all $S\subset[m]$ (so they are maximally-far from any $\F_2^n$-linear function).
We note that Maiorana--McFarland functions were also recently used to prove an exponential separation between classical query complexity and a single quantum query in the strict regime of zero-error quantum computation \cite{GirishServedio2025ForrelationExtremallyHard}.

We begin by proving hardness of testing $\mathrm{MM}_n$ using only classical samples $\big( x, f(x)\big)$ and Fourier samples.

\begin{theorem}\label{thm:fourier-sampling-not-sufficient-for-mm}
    Suppose a tester for $\mathrm{MM}_n$ using exclusively classical samples and Fourier samples succeeds with probability more than $1/2+2^{-0.7n}$ for the accuracy parameter $\epsilon=1/2-2^{-0.3n}$.
    Then the tester uses at least $2^{0.1n}$ classical samples.
\end{theorem}
\begin{proof}
    Define $\mc H$ be the set of Boolean functions on $n$ bits with bias at most $2^{-n/3}$.
    From Chernoff we have that for uniformly random $f:\F_2^n\to\F_2$, the bias almost always satisfies this bound:
    \begin{equation}
        \label{eq:correlation-bound-f}
        \Pr_f\Big[\;\Big|\E_x\left[ (-1)^{f(x)}\right]\Big| \geq 2^{-n/3} \;\Big] \leq \exp(-\mathsf{const}\cdot 2^{n/3})\,.
    \end{equation}
    Now consider the two ensembles of Boolean functions $\F_2^n\times \F_2^n\to\F_2$
    \[\mc F_1 = \big\{(x,y)\mapsto \langle x,y\rangle + h(x)\big\}_{h\sim \mc H}\qquad\text{and}\qquad \mc F_2 = \big\{(x,y)\mapsto \langle x,y\rangle + h(y)\big\}_{h\sim \mc H}\,,\]
    Note that for any $f=\langle x,y\rangle + h(x)\in\mathrm{MM}_n$ and any $g=\langle x,y\rangle + m(y)\in\supp\mc F_2$,
    \begin{align*}
        \Big|\E_{x,y}(-1)^{f(x,y)}(-1)^{g(x,y)}\Big|=\Big|\E_{x,y}(-1)^{h(x)}(-1)^{m(y)}\Big|=|\mathrm{bias}(h)\mathrm{bias}(m)|\leq 2^{-n/3}\,.
    \end{align*}
    Thus all $g\in \supp \mc F_2$ are at least $(1/2-2^{-n/3})$-far from $\mathrm{MM}_n\supseteq \supp \mc F_1$.

    Suppose a testing algorithm $A$ using $R\leq 2^{0.1n}$ classical samples succeeds with probability $\delta$.
    This means that, given access to a function drawn uniformly at random from $\mc F_1$ or from $\mc F_2$ with equal probability, the algorithm $A$ can guess which of the two ensembles the function was drawn from with success probability at least $\delta$.

    Define $\widetilde{\mc F}_1$ (resp. $\widetilde{\mc F}_2$) to be $\mc F_1$ (resp. $\mc F_1$) but where $h$ is drawn uniformly at random from all Boolean functions.
    Because of the bound \eqref{eq:correlation-bound-f}, we know
    \[\left|\Pr_{f\sim\widetilde{\mc F}_1}[A \text{ accepts } f] - \Pr_{f\sim \mc F_1}[A \text{ accepts } f]\right|\leq \exp(-\mathsf{const}\cdot 2^{-n/3}),\]
    and similarly for $\mc F_2, \widetilde{\mc F}_2$.
    This means that in the $\widetilde{\mc F}$ version of the distinguishing game, $A$ succeeds with probability at least $\delta - \exp(-\mathsf{const}\cdot 2^{-n/3})$.

    Let $b\sim\{1,2\}$ and $f\sim\mc F_b$.
    With probability at least $1-{R^2}/2^n$, all $x^{(r)}$'s and all $y^{(r)}$'s in the $R$-many samples are distinct (collision bound and union bound).
    Call this event $D$.
    Conditioned on $D$, the distribution of observed values $\big(f(x^{(1)},y^{(1)}),\ldots, f(x^{(R)},y^{(R)})\big)$ is uniformly random because in $\widetilde{\mc F}_b$, $h$ is a uniformly random Boolean function.
    Thus conditioned on $D$, the data observed is independent of $b$.
    Moreover, all functions under consideration are bent, so Fourier sampling provides no information whatsoever.
    The distinguishing probability is thus bounded by:
    \[\delta \leq \Pr[D]\cdot \frac12 + \Pr[D^c]\cdot 1 + e^{-\mathsf{const}\cdot 2^{-n/3}} \leq \frac12 + \frac{R^2}{2^n}+ e^{-\mathsf{const}\cdot 2^{-n/3}} \leq \frac{1}{2} + 2^{-0.7n}\]
    for $n$ large enough.
\end{proof}

Having established that the Maiorana--McFarland class is hard to test from classical samples and Fourier samples alone, we now give a very efficient passive quantum tester for $\mathrm{MM}_n$.
While this tester still uses quantum Fourier sampling at the end of the algorithm, it crucially preprocesses the function state in superposition before applying performing Fourier sampling.

\begin{theorem}\label{thm:quantum-mm-tester}
    There is an efficient quantum algorithm that uses $\mathcal{O}(1)$ copies of the function state $\ket{f}=\frac{1}{2^n}\sum_{x,y\in\{0,1\}^n}\ket{x,y,f(x,y)}$ to decide, with success probability $\geq 2/3$, whether $f$ is in $\mathrm{MM}_n$ or $(1/3)$-far from $\mathrm{MM}_n$.
\end{theorem}
\begin{proof}
    Let $U$ denote the $(2n+1)$-qubit unitary acting as $\ket{x,y,b}\mapsto \ket{x,y,b\oplus \langle x,y\rangle}$.
    Note that $U$ can be implemented by a quantum circuit with $\mathcal{O}(n)$ many $2$-qubit gates and depth $\mc O(\log n)$.
    Moreover, notice that $U\ket{f}=\ket{\tilde{f}}$ for $\tilde{f}(x,y):=f(x,y)\oplus \langle x,y\rangle$.

    The quantum algorithm works as follows. 
    Recall that one function state copy suffices to obtain one Fourier sample with success probability $1/2$, using only $2n$ many single-qubit gates.  
    Applying this Fourier sampling subroutine to $\mathcal{O}(1)$ many copies of $U\ket{f}$ thus suffices to obtain, with success probability $\geq 5/6$, $m \geq  C$  many Fourier samples $S_1,\ldots, S_m\subseteq[2n]$ of the function $\tilde{f}$, where $C>0$ is a universal constant to be chosen later.
    Let $J=\{n+1,n+2,\ldots, n\}$ and compute
    \[\hat{p}= \frac1m\left|\big\{1\leq k\leq m: J\cap S_k\neq\emptyset\big\}\right|\,.\]
    If $\hat{p}\leq 1/9$, output ``$f\in \mathrm{MM}_n$''. Otherwise, output ``$f$ is $(1/3)$-far from $\mathrm{MM}_n$.''

    First, let us show completeness of the protocol. So, suppose $f\in\mathrm{MM}_n$. Then there is a function $h:\mathbb{F}_2^n\to\mathbb{F}_2$ such that 
    \begin{equation*}
        U\ket{f}
        = \frac{1}{2^n}\sum_{x,y}\ket{x,y,h(x)} = \ket{h} \, ,
    \end{equation*}
    where we abused notation by using $h$ to denote the function $h:\mathbb{F}_2^n\times \mathbb{F}_2^n\to\mathbb{F}_2$ defined as $h(x,y)=h(x)$. 
    As $h(x,y)$ depends only on the first $n$ variables, we have
    \[p\coloneqq \sum_{S\subseteq [2n],\, J\cap S\neq \emptyset} |\hat{h}(S)|^2 = 0\,.\]
    The constant $C$ can be chosen such that, conditioned on the high probability event that we obtained at least $C$ many Fourier samples, we have $|p-\hat{p}|\leq 1/9$ with probability $\geq 5/6$ (by Chernoff-Hoeffding). 
    So $\hat{p}\leq 1/9$, and our tester correctly outputs ``$f\in\mathrm{MM}_n$'' with probability $\geq 2/3$.

    Next, we analyze soundness. So, suppose $f$ is $(1/3)$-far from $\mathrm{MM}_n$. Equivalently, $\tilde{f}(x,y)=f(x,y)\oplus\langle x,y\rangle$ is $(1/3)$-far from any Boolean function $h$ that depends only on the first $n$ variables, $h(x,y)=h(x)$, where we again abused notation.
    Consider the function $g$ defined as \[g(x,y)=\sum_{S\subseteq [2n]:J\cap S=\emptyset}\hat{\tilde{f}}(S)\chi_S(x,y)\, .\]
    Notice that $g(x,y)$ depends only on $x$, but $g$ is in general not Boolean.
    Define $\tilde{g}(x,y)=\mathbf{1}_{g(x,y)\geq 1/2}$. Notice that $\tilde{g}$ is a Boolean function and that $\tilde{g}(x,y)$ depends only on $x$. Then (compare \cite[Fact II.2]{atici2007quantum}) we have
    \begin{align*}
        \frac{1}{3}
        \leq \mathbb{P}_{x_1,x_2}[\tilde{f}(x_1,x_2)\neq \tilde{g}_2(x_2)]
        \leq \mathbb{E}_{x_1,x_2}[(\tilde{f}(x_1,x_2) - g(x_1,x_2))^2]
        = \sum_{S\subset [2n]: J\cap S\neq \emptyset}\left\lvert \hat{\tilde{f}}(S)\right\rvert^2
        = p\, .
    \end{align*}
    Again, conditioned on having produced at least $C$ many Fourier samples, with probability $\geq 5/6$, we have $|p-\hat{p}|\leq 1/9$ and thus $\hat{p}\geq 2/9$. So, our tester correctly outputs ``$f$ is $(1/3)$-far from $\mathrm{MM}_n$'' with probability $\geq 2/3$.
\end{proof}

\begin{remark}
    We note that the above construction can be modified to give a limitation on non-adaptive quantum learning algorithms that use only Fourier samples and classical samples.
    For $\mathcal{H}$ an arbitrary class of Boolean functions on $n$ bits, define the concept class $\mathcal{F}=\mathcal{F}(\mathcal{H})=\{f_h\}_{h\in\mathcal{H}}$, again with $f_h(x,y)=\langle x,y \rangle + h(x)$ as above.
    If we can learn $\mathcal{H}$ efficiently from function states, then we can also learn the class $\mathcal{F}$ with only a constant overhead in the number of copies and a linear overhead in quantum time complexity. 
    To achieve this, apply the quantum learning algorithm for $\mathcal{H}$ on copies of the suitably transformed function state $U\ket{f}$ with $U$ as in the proof of \Cref{thm:quantum-mm-tester}. 
    In contrast, consider a learner that only uses Fourier sampling and classical samples. 
    Because the class 
    $\mathcal{F}$ consists of bent functions, Fourier sampling provides no information. 
    Furthermore, learning an unknown $f\in \mathcal{F}$ is at least as hard as learning an unknown $h \in \mathcal{H}$ from classical samples as any labeled sample for $f$ can be converted to a labeled sample for $h$. 
    Therefore, when limiting the learner to Fourier samples and classical samples, it inherits classical sample complexity lower bounds for learning $\mathcal{H}$. This implies that whenever function state copies offer a quantum advantage over classical samples in learning $\mathcal{H}$ (e.g., for $k$-juntas~\cite{chen2023testing} and for degree-$d$ polynomials over $\mathbb{F}_2$~\cite{arunachalam2023optimal}), we obtain a separation between quantum learning and
    quantum learning using only Fourier sampling and classical samples.
\end{remark}

\section{Separating passive quantum from query-based classical property testing}\label{section:separating-classical-queries-from-function-states}

In this section we give a property for which classical queries have exponential advantage over quantum testing from function states.
This property is closely related to the inability of quantum computers to measure the intersection of three subset states, where for a subset $S\subseteq \mathbb{F}_2^n$, the corresponding subset state is defined as $\ket{S}=\frac{1}{\sqrt{|S|}}\sum_{x\in S}\ket{x}$. We explain this connection at the end of the section.

The main result of this section is the following theorem.

\begin{theorem}
    \label{thm:qsamp-cquery-sep}
    There exist two sets of Boolean functions $F_0,F_1$ such that
    \[\min_{f_0\in F_0,f_1\in F_1} \|f_0-f_1\|_1\geq \frac{1}{64}\]
    and such that:
    \begin{itemize}
        \item Any passive quantum tester requires $\Omega(2^{n/2})$ copies of a function state to distinguish $F_0$ and $F_1$ with constant probability $2/3$.
        \item $F_0$ and $F_1$ may be distinguished with probability $2/3$ from $\mc O(1)$ classical queries.
    \end{itemize}
\end{theorem}

The families $F_0$ and $F_1$ arise from certain encodings of triples of subsets $A,B,C\subseteq\{0,1\}^n$.
Consider the class of Boolean functions $f_{(A,B,C)}:\{0,1\}^{n+2}\to\{0,1\}$ on $n+2$ bits parameterized by subsets $A,B,C\subseteq\{0,1\}^n$ and defined as follows:
\[f_{(A,B,C)}(x,a)=\begin{cases}
    \mathbf{1}_{x\in A} &  a = 00\\
    \mathbf{1}_{x\in B} &  a = 01\\
    \mathbf{1}_{x\in C} &  a = 10\\
    0 &  a = 11
\end{cases}\, .\]
With $A,B,C$ drawn uniformly from subsets of $\{0,1\}^n$, we define two function state ensembles $\{\ket{f_{(A,B,C)}}\}_{A,B,C}$ and $\{\ket{f_{(A,B,A\Delta B)}}\}_{A,B}$, with their mixed state average over $t$-copy states given by
\begin{align*}
    \mc E_0 = \E_{A,B,C} \left[\Ket{f_{(A,B,C)}}\Bra{f_{(A,B,C)}}^{\otimes t}\right],\qquad 
    \mc E_1 = \E_{A,B} \left[\Ket{f_{(A,B,A\Delta B)}}\Bra{f_{(A,B,A\Delta B)}}^{\otimes t}\right]\, .
\end{align*}

We now show that these two mixed state averages over $t$-copy states are close in trace distance unless $t$ scales exponentially in $n$. This means that exponentially-in-$n$ many copies are needed to distinguish between the two function state ensembles.

\begin{theorem}
    \label{thm:subset-indist}
    $\|\mc E_0 - \mc E_1\|_1 \leq \mc O(t/2^{n/2}).$
\end{theorem}
\begin{proof}
    It will help to reinterpret $\Ket{f_{(A,B,C)}}$ as a subset state via the rewriting
    \begin{equation}
    \label{eq:fnc-to-subset-state}
    \ket{x,a}\ket{f_{(A,B,C)}(x,a)}=\sum_{b\in\{0,1\}}\mathbf 1_{x\in S_{a,b}}\ket{x,(a,b)},
    \end{equation}
    where $S_{a,b}$ denotes $A,B,C,$ or $\emptyset$ according to $a$ when $b=1$, or the respective complements if $b=0$.
    Using $r$ to represent the concatenation of $a$ and $b$ we may then write
    \[\ket{f_{(A,B,C)}}=\frac{1}{\sqrt{4N}}\sum_{r\in\{0,1\}^3}\sum_{x\in\{0,1\}^n}\mathbf{1}_{x\in S_r}\ket{x,r},\]
    where $N:=2^n$ and, like before, $S_r$ denotes one of $A,B,C,$ or $\emptyset$ or the complements thereof.
    
    With this notation let us consider the basis for the space of $t$ copies of function states given by
    \[\big\{\ket{x_1,r_1,\ldots, x_t, r_t}:x_j\in\{0,1\}^n, r_j\in\{0,1\}^3, j=1,\ldots, t\big\}\,.\]
    Let $\Pi$ denote the projector onto the subspace spanned by those $\ket{x_1,r_1,\ldots, x_t, r_t}$ for which all $x_j$ are distinct.

    First, we claim
    \begin{equation}
    \label{eq:distclose}
        \|\mc E_0 - \Pi \mc E_0 \Pi\|_1,\,\|\mc E_1-\Pi\mc E_1\Pi\|_1\; \leq \; \mc O\left(\frac{t}{\sqrt{N}}\right)\,.
    \end{equation}
    These bounds follow from applying the triangle inequality to the following estimate:
    for any fixed $A,B,C$, we have
    \begin{align}
    \label{eq:ens-pw-est}
        \big\|\Ket{f_{(A,B,C)}}\Bra{f_{(A,B,C)}}^{\otimes t} - \Pi\Ket{f_{(A,B,C)}}\Bra{f_{(A,B,C)}}^{\otimes t}\Pi\big\|_1 &=\sqrt{\left(1+3\frac{4^tN^{\underline t}}{(4N)^t}\right)\left(1-\frac{4^tN^{\underline t}}{(4N)^t}\right)}\\
        &\leq \frac{2t}{\sqrt{N}}\, ,
    \end{align}
    where $(x)^{\underline t}=x(x-1)\ldots (x-t+1)$ denotes falling factorial, and where in the second step we applied the bound
    \[\frac{4^tN^{\underline t}}{(4N)^t}\geq \left(1-\frac{t}{N}\right)^{t}\geq 1-\frac{t^2}{N}.\]
    To see \Cref{eq:ens-pw-est}, note that
    \[M := \Ket{f_{(A,B,C)}}\Bra{f_{(A,B,C)}}^{\otimes t} - \Pi\Ket{f_{(A,B,C)}}\Bra{f_{(A,B,C)}}^{\otimes t}\Pi\]
    has the following block form after reordering columns:
    \vspace{-1em}
\[
\NiceMatrixOptions{cell-space-limits = 10pt}
        M=\frac{1}{(4N)^t}\begin{pNiceArray}{c|c|c}[first-row]
         & &  \\
        \;\hspace{1em}0\hspace{1em}\; & \;\hspace{1em}1\hspace{1em}\; & \hspace{1em}0\hspace{1em}\; \\
        \hline
        \;\hspace{1em}1\hspace{1em}\; & \;\hspace{1em}1\hspace{1em}\; & \hspace{1em}0\hspace{1em}\; \\
        \hline
        \;\hspace{1em}0\hspace{1em}\; & \;\hspace{1em}0\hspace{1em}\; & \hspace{1em}0\hspace{1em}\; \\
        \CodeAfter
        \UnderBrace[ yshift=6pt]{2-1}{3-1}{\scriptstyle 4^tN^{\underline t}}
        \UnderBrace[ yshift=6pt]{2-2}{3-2}{\scriptstyle 4^t(N^t-N^{\underline t})}
        \UnderBrace[ yshift=6pt]{2-3}{3-3}{\scriptstyle \hspace{2em}(8N)^t-(4N)^t}
        \CodeAfter
   \tikz{
   \draw [decorate,decoration = {calligraphic brace,amplitude=4pt},line width=1.25pt]
        ([xshift=4mm]1-3.north east) to node [auto = left] {\;$\scriptstyle 4^tN^{\underline t}$} 
        ([xshift=4mm]1-3.south east);
        \draw [decorate,decoration = {calligraphic brace,amplitude=4pt},line width=1.25pt]
        ([xshift=4mm]2-3.north east) to node [auto = left] {\;$\scriptstyle 4^t(N^t-N^{\underline t})$} 
        ([xshift=4mm]2-3.south east);
        \draw [decorate,decoration = {calligraphic brace,amplitude=4pt},line width=1.25pt]
        ([xshift=4mm]3-3.north east) to node [auto = left] {\;$\scriptstyle(8N)^t-(4N)^t$} 
        ([xshift=4mm]3-3.south east);
        }
        \end{pNiceArray}\hspace{5.5em}.
    \]
    \vspace{2em}

    \noindent This is because $\Ket{f_{(A,B,C)}}\Bra{f_{(A,B,C)}}^{\otimes t}$
    is an all-zeros matrix except for the principal submatrix associated to indices $\big((x_1,r_1),\ldots, (x_t,r_t)\big)$ where $x_j\in S_{r_j}$ for all $j$, and here it is equal to $(4N)^{-t}$.
    There are $(4N)^t$ such entries.
    Moreover, $\Pi\Ket{f_{(A,B,C)}}\Bra{f_{(A,B,C)}}^{\otimes t}\Pi$ is an all-zeros matrix except for the principal submatrix associated to indices $\big((x_1,r_1),\ldots, (x_t,r_t)\big)$ where $x_j\in S_{r_j}$ for all $j$ and $x_j\neq x_k$ for $j\neq k$ and here it is also equal to $(4N)^{-t}$---and there are $4^tN^{\underline t}$ of these entries.
    $M$ thus has rank 2 and its spectrum is easily determined, leading to the estimate in \Cref{eq:ens-pw-est}.

    Now we claim that in fact
    \begin{equation}
    \label{eq:disteq}
    \Pi \mc E_0 \Pi = \Pi \mc E_1 \Pi\,.
    \end{equation}
    Let us consider a specific entry in $\Pi \mc E_0 \Pi$ with row and column indices
    \[\mb r = (\ldots, (x_j, r_j),\ldots), \qquad \mb s= (\ldots, (y_j,s_j),\ldots)\,.\]
    It will be useful to write $\mc S_z=\mc S_z(\mathbf{r},\mathbf{s})$ for the set types that appear with a particular string $z\in\{0,1\}$ in $\mb r$ and $\mb s$.
    That is, for any $z\in\{0,1\}^n$ define
    \[\mc S_z = \mc S_z(\mathbf{r},\mathbf{s})=\big\{q\in\{0,1\}^3 : (z,q)\in \mb r \text{ or } (z,q)\in \mb s\big\}.\]
    Then \[\bra{\mb r}\Pi \mc E_0 \Pi\ket{\mb s} = (4N)^{-t}\E_{A,B,C}\left(\tprod_j\mb 1_{x_j\in S_{r_j}}\right)\left(\tprod_j\mb 1_{y_j\in S_{s_j}}\right)=(4N)^{-t}\prod_{z\in \{x_j\}_j\cup\{y_j\}_j}\E_{A,B,C}\tprod_{q\in \mc S_z}\mathbf{1}_{z\in S_q}.\]
    It follows from the definition of $\Pi$ that $|\mc S_z|\leq 2$ for any $z\in \{x_j\}_j\cup \{y_j\}_j$: there is at most one contribution to $\mc S_z$ from each of $\mb r$ and $\mb s$.
    As a result we have for any $z$ that
    \[\E_{A,B,C\overset{\mathsf{iid}}{\sim} \mc P\{0,1\}^n}\tprod_{q\in \mc S_z}\mathbf{1}_{z\in S_q}=\E_{\substack{A,B\overset{\mathsf{iid}}{\sim} \mc P\{0,1\}^n\\C=A\Delta B}}\tprod_{q\in \mc S_z}\mathbf{1}_{z\in S_q}.\]
    This follows from mild case analysis, the most important part of which is to note that for any $S\neq T\in \{A,B,C,A\Delta B\}$,
    \[(\mb 1_{x\in S},\mb 1_{x\in T}) 
    \sim (b_1,b_2),\]
    where $b_1$ and $b_2$ are i.i.d. Bernoulli $1/2$ random variables.
    So we see $\bra{\mb r}\Pi \mc E_0 \Pi\ket{\mb s}=\bra{\mb r}\Pi \mc E_1 \Pi\ket{\mb s}$ and \Cref{eq:disteq} is satisfied.

    Combining the triangle inequality with \Cref{eq:distclose} and \Cref{eq:disteq} gives the result.
\end{proof}

\begin{proof}[Proof of Theorem \ref{thm:qsamp-cquery-sep}]
    Consider
    \begin{align*}
        &F_0 = \big\{f_{(A,B,C)}:A,B,C\subseteq\{0,1\}^n, 2^{-n}|A\cap B\cap C|\geq 1/16\big\}\\[0.5em]
        \text{and}\quad &F_1 = \big\{f_{(A,B,A\Delta B)}: A,B\subseteq\{0,1\}^n\big\},
    \end{align*}

    First we prove the minimum distance between $F_0$ and $F_1$.
    For any $f_0\in F_0$, there are $2^n/16$ strings $x\in\{0,1\}^n$ such that $f_0(x00)=f_0(x01)=f_0(x10)=1$.
    On the other hand, for all $f_1\in F_1$, by definition there are no strings $x$ with this property.
    Thus the minimum $L_1$ distance between $F_0$ and $F_1$ is at least
    \[\frac{2^n/16}{4\cdot2^{n}}=\frac{1}{64}\,.\]
    
    Now define the state ensembles
    \begin{align*}
        \mc E_0' = \E_{f\sim F_0}\ketbra{f}{f}^{\otimes t}\quad\text{and}\quad\mc E_1 = \E_{f\sim F_1}\ketbra{f}{f}^{\otimes t}.
    \end{align*}
    $\mc E_1$ here is exactly $\mc E_1$ from Theorem \ref{thm:subset-indist}.
    To compare $\mc E_0'$ and $\mc E_0$ from Theorem \ref{thm:subset-indist}, note that for $A,B,C\overset{\mathsf{iid}}{\sim}\mc P\{0,1\}^n$, any string $x$ is in $A\cap B\cap C$ with probability $1/8$ and so from Chernoff we have
    \[\Pr\left[|A\cap B \cap C| < \frac12\cdot\frac{2^n}{8}\right]\leq \exp\left(-\frac{2^n}{64}\right).\]
    This dramatic concentration, together with Theorem \ref{thm:subset-indist} implies
    \[\|\mc E_0'-\mc E_1\|_1\leq \|\mc E_0'-\mc E_0\|_1+\|\mc E_1-\mc E_0\|_1\leq \mc O(t/2^{n/2})\, .\]
    
    To test this property with classical queries, given an unknown $f=f_{(A,B,C)}$ one may simply choose a random $x\in\{0,1\}^n$ and check if $f(x00)=f(x01)=f(x10)=1$.
    This test accepts with probability $1/8$ when $f\in F_0$ and accepts with probability $0$ when $f\in F_1$.
\end{proof}

\subsection{\boldmath $k$-fold intersection is ``unfeelable'' for $k \geq 3$}

In this subsection, we reinterpret \Cref{thm:qsamp-cquery-sep} in the context of subset states. 
Given access to copies of $k$ different subset states $\ket{S_1},\ldots, \ket{S_k}$, it is natural to ask how many copies of each are required to estimate the fractional size of the mutual intersection,
\[\frac{|S_1\cap\cdots\cap S_k|}{2^n}.\]
When $k=2$, this can be readily accomplished using ideas similar to our algorithms presented above.
In the case of intersection estimation with $k=2$, we have the identity
\[\frac{|S_1\cap S_2|}{2^n}=\braket{S_\text{all}|S_1}\braket{S_1|S_2}\braket{S_2|S_\text{all}},\]
where $S_\text{all}:=\{0,1\}^n$ denotes the full hypercube.
The quantities on the right-hand side are easily estimated via swap tests, so it takes $\mc O(1)$ copies of $\ket{S_1},\ket{S_2}$ to estimate the quantity of interest to any constant additive error.

In contrast, it is a consequence of Theorem \ref{thm:qsamp-cquery-sep} that the same question for $k=3$ has a very different answer: it requires $\Omega(2^{n/2})$ copies to achieve constant additive error.
To see this, note that from any $\ket{f_{(A,B,C)}}$ one may obtain each of $\ket{A}$, $\ket{B}$, and $\ket{C}$ with constant probability by measuring the $a$ and $f(x,a)$ registers, provided that the minimum among $|A|$, $|B|$, and $|C|$ is at least a constant fraction of $2^n$---and this condition is satisfied by the overwhelming majority of functions in the families $F_0$ and $F_1$ of Theorem \ref{thm:qsamp-cquery-sep}.
From $F_0$ and $F_1$ we obtain:
\begin{corollary}
    There are two families $\mc S_0$ and $\mc S_1$ of triples of subsets of $\{0,1\}^n$ such that
    \begin{align*}
    &\forall(A_0,B_0,C_0)\in \mc S_0, \quad |A_0\cap B_0\cap C_0|/2^n \geq 1/16\\[0.5em]
    \text{ and }\qquad&\forall(A_1,B_1,C_1)\in \mc S_1,\quad |A_1\cap B_1\cap C_1|/2^n =0\,,
    \end{align*}
    and yet any quantum algorithm distinguishing the two families via their subset states requires $\Omega(2^{n/2})$ copies of $\ket{A},\ket{B},$ or $\ket{C}$.
\end{corollary}

\subsection{2-fold intersection detection using quantum Fourier sampling}\label{sec:two-fold-intersection}

As noted above, in contrast to 3-fold intersection detection, 2-fold intersection detection is quantumly easy. In this subsection, we show that 2-fold intersection detection from function states can even be achieved with quantum Fourier sampling by exactly estimating the intersection $|A \cap B|/2^n$.

Consider the class of Boolean functions $f_{(A,B)}:\{0,1\}^{n+1}\to\{0,1\}$ on $n+1$ bits parameterized by subsets $A,B\subseteq\{0,1\}^n$ and defined as 
\[f_{(A,B)}(x, a)=\begin{cases}
    \mathbf{1}_{x\in A} &  a = 0\\
    \mathbf{1}_{x\in B} &  a = 1\\
\end{cases} \, .\]

\begin{theorem}[2-fold intersection estimation] 
    There is an efficient 
    quantum algorithm that uses unentangled measurements on $\mathcal{O}\left(\frac{\log(1/\delta)}{\varepsilon^2}\right)$ many independent copies of the function state $\ket{\Psi}=\frac{1}{\sqrt{2^{(n+1)}}}\sum_{x\in\{0,1\}^{n+1}}\ket{x,f_{(A,B)}(x)}$ and estimates $|A \cap B|/2^n$ to additive error $\varepsilon$ with success probability $\geq 1-\delta$. 
\end{theorem}

\begin{proof}
    For ease of notation, we write $f=f_{(A,B)}$ throughout the proof.
    First, we will show how to estimate a particular linear combination of set sizes $|A|, |B|$ and $|A \cap B|$ using quantum Fourier sampling. 
    As is common in Boolean Fourier analysis, relabel the output of our function  $0 \mapsto 1$ and $1 \mapsto -1$. 
    Consider the expression $\mathbb{E}_x \left[ f(x, 0) \cdot f(x, 1) \right]$. 
    Simple case analysis reveals that this expression equals $\frac{1}{2^n} \left( 4|A \cap B| - 2|A| - 2|B| + 2^n \right)$.
    
    Furthermore, we can rewrite $\mathbb{E}_x \left[ f(x, 0) \cdot f(x, 1) \right]=1 - 2 \cdot \Pr[f(x,0) \ne f(x,1)]=1 - 2 \cdot \mathbf{Inf}_{n+1} [f]$ (see \cite{odonnell2021analysisbooleanfunctions}, Definition 2.13). 
    Using that $\mathbf{Inf}_{n+1} [f] = \sum_{S\ni n+1}\widehat{f}(S)^2$ (see \cite{odonnell2021analysisbooleanfunctions}, Theorem 2.20), we can estimate $\mathbf{Inf}_{n+1} [f]$ as the probability that quantum Fourier sampling produces a set that contains $n+1$.
    By Chernoff-Hoeffding, we can estimate this quantity to additive error $\frac{2\varepsilon}{3}$ and success probability $\frac{\delta}{3}$ using $\mathcal{O}(\log(1/\delta)/\varepsilon^2)$ copies of our state (compare the discussion around \Cref{footnote:quantum_fourier_sampling}).
    
    By standard Chernoff-Hoeffding bounds, the fractional sizes of $A$ and $B$ can easily be estimated to additive error $\frac{2\varepsilon}{3}$. 
    A union bound over all components implies that the resulting overall estimate obtained for $|A \cap B|/2^n$ has additive error $\varepsilon$ with success probability $\geq 1 - \delta$, and the result follows. 
\end{proof}

\section{A challenge: lower bounds for monotonicity testing}\label{section:passive-quantum-monotonicity-testing-lower-bound}


Here we show that the ensembles used in \cite{Goldreich2000} to establish strong lower bounds on monotonicity testing from samples do not improve upon the basic $\Omega(1/\varepsilon)$ sample complexity lower bound in the quantum case. 
%
%
%
To prove this, we consider the pair of distributions over functions from \cite{Goldreich2000}, constructed such that one is supported entirely on monotone functions, and the other with high probability on functions that are $\epsilon$-far from monotone; we show that the associated $t$-copy quantum function state ensembles become distinguishable with constant success probability as soon as $t=\Omega(1/\varepsilon)$.
At the end of the section, we discuss how to extend our reasoning to the ensembles used in \cite{black_nearly_2023}.

\subsection{Distinguishability of twin ensembles}

For the proof, it will be useful to also consider \emph{phase states}, which are given by
\begin{equation*}
    \ket{\Psi_{f}^{\mathrm{ph}}}
    = \frac{1}{\sqrt{2^n}}\sum_{x\in\{0,1\}^n} (-1)^{f(x)}\ket{x}\, .
\end{equation*}
The proof reduces to the distinguishability of phase state 
ensembles encoding the following classical sets of functions taken from \cite{Goldreich2000}. 

\begin{definition}[Twin ensembles]
    Let $M = \{(u_i,v_i)\}_{i=1}^m$ be a set of pairs of elements in $\{0,1\}^n$ such that all $u_1,v_1,\ldots, u_m,v_m$ are distinct.
Let $\cup M := \cup_{(u,v)\in M}\{u,v\}$ be the complete set of elements in the matching.
Fix a function ${g:\{0,1\}^n\backslash(\cup M)\to\{0,1\}}$.
We now define the \emph{twin ensembles} associated to $M$ and $g$, which are two sets $F_0, F_1$ of functions on $\{0,1\}^n$.

For any bipartition of $M$, $M=A\sqcup B$, define the following two functions.
\begin{enumerate}
    \item $f_{A,B}^{(0)}$ is defined as follows:
    \begin{itemize}
        \item For $(u,v)\in A$, we set $f_{A,B}^{(0)}(u) = f_{A,B}^{(0)}(v)=1$.
        \item For $(u,v)\in B$, we set $f_{A,B}^{(0)}(u) = f_{A,B}^{(0)}(v)=0$.
        \item If $x\not\in \cup M$, then define $f_{A,B}^{(0)}(x)=g(x).$
    \end{itemize}
    \item $f_{A,B}^{(1)}$ is defined as follows:
    \begin{itemize}
        \item For $(u,v)\in A$, we set $f_{A,B}^{(1)}(u) = 1$ and $f_{A,B}^{(1)}(v)=0$.
        \item For $(u,v)\in B$, we set $f_{A,B}^{(1)}(u) = 0$ and $f_{A,B}^{(1)}(v)=1$.
        \item If $x\not\in \cup M$, then define $f_{A,B}^{(1)}(x)=g(x).$
    \end{itemize}    
\end{enumerate}
Then the twin ensembles associated to $M$ and $g$ are $F_0 = \{f_{A,B}^{(0)}\}_{A\sqcup B = M}$ and $F_1 = \{f_{A,B}^{(1)}\}_{A\sqcup B = M}$. 
\end{definition}

Let us recall the reasoning from \cite{Goldreich2000} that connects these ensembles to monotonicity testing.
    Take $k=\lceil n/2\rceil$ and consider the $k^{\mathrm{th}}$ and $(k-1)^{\mathrm{th}}$ layer of the Boolean hypercube with respect to the standard partial ordering on strings $\preceq$.
    These layers we denote by $L_k$ and $L_{k-1}$ respectively; \textit{i.e.}, $L_i=\{x\in\{0,1\}^{n-1}: |x|=i\}$, where $|x|$ denotes the Hamming weight of $x$.
    Stirling's formula gives that $|L_k|, |L_{k-1}|=\Omega(2^{n}/\sqrt{n})$.
     As argued in \cite{Goldreich2000}, we can find a matching $M=\{(u_i\prec v_i)\}_{i=1}^m \subset L_{k-1}\times L_k$ such that (\textit{i}) there is no $i\neq j$ such that $u_i$ and $v_j$ are comparable, 
     (\textit{ii}) $|M|$ is even, 
     and (\textit{iii}) $m\coloneqq |M|=\varepsilon \cdot 2^{n}$, for any $\epsilon=\epsilon(n)$ with $0<\epsilon\leq \mc O(n^{-3/2})$.
     Now define
     \begin{align*}
         g:\{0,1\}^n\backslash{\cup M} &\to \{0,1\}\\
         x &\mapsto \mathbf{1}_{|x|\geq n/2}\,.
     \end{align*}
     The choices of $g$ and $M$ define twin ensembles $F_0$ and $F_1$.

    Clearly, every $f_{A,B}^{(0)}$ is a monotone function.
    Let $F_1^\text{far}\subset F_1$ be the set of $f_{A,B}^{(1)}$ functions for which $|B|\geq m/4$.
    Then all functions in $F_1^\text{far}$ are at least $\Omega(\epsilon)$-far from monotone \cite{Goldreich2000}.
    We thus wish to bound the distinguishability of $F_0$ from $F_1^{\text{far}}$, which in the quantum case is determined by the $1$-norm    
    \[\left\|\E_{f_{A,B}^{(0)}\sim F_0}\left(\ketbra{f_{A,B}^{(0)}}{f_{A,B}^{(0)}}\right)^{\otimes t} - \E_{f_{A,B}^{(1)}\sim F^\text{far}_1}\left(\ketbra{f_{A,B}^{(1)}}{f_{A,B}^{(1)}}\right)^{\otimes t}\right\|_1\]
    From a standard concentration argument, it suffices to instead bound the distinguishability between $F_0$ and $F_1$, which in the quantum case is determined by 
    \[\left\|\E_{f_{A,B}^{(0)}\sim F_0}\left(\ketbra{f_{A,B}^{(0)}}{f_{A,B}^{(0)}}\right)^{\otimes t} - \E_{f_{A,B}^{(1)}\sim F_1}\left(\ketbra{f_{A,B}^{(1)}}{f_{A,B}^{(1)}}\right)^{\otimes t}\right\|_1\, .\]

We will show that, in contrast to the classical case analyzed in \cite{Goldreich2000}, the twin ensembles actually become distinguishable already for $t\sim 1/\varepsilon$. Namely, much of \Cref{section:passive-quantum-monotonicity-testing-lower-bound} will be dedicated to proving the following theorem: 

\begin{theorem}
    \label{thm:distinguishability-pair-ens}
    Define $\epsilon>0$ so that $\epsilon 2^n = m = |M|$.
    Then 
    \begin{equation}
    \label{eq:tr-dist-f0-f1}
    \left\|\mathbb{E}_{B}\left[\left(\ketbra{\Psi_{f_{A,B}^{(0)}}^{\textnormal{ph}}}{\Psi_{f_{A,B}^{(0)}}^{\textnormal{ph}}}\right)^{\otimes t}\right]-\mathbb{E}_{B}\left[\left(\ketbra{\Psi_{f_{A,B}^{(1)}}^{\textnormal{ph}}}{\Psi_{f_{A,B}^{(1)}}^{\textnormal{ph}}}\right)^{\otimes t}\right]\right\|_1 \geq \Omega(1)
    \end{equation}
    for $t = \Omega(1/\varepsilon)$.
\end{theorem}

Let us note that \Cref{thm:distinguishability-pair-ens} implies the same $1$-norm lower bound and thus the same distinguishability of the two ensembles also from function state copies.
     To see this, we notice that the function state for any Boolean function $f$ defined on $n$ bits is unitarily equivalent to the phase state for a related Boolean function on $n+1$ bits: 

     \[ ( I^{\otimes n}\otimes H) \ket{f}
     = \frac{1}{\sqrt{2^{n+1}}}\sum_{x\in\{0,1\}^n}\ket{x}\ket{0} + \frac{1}{\sqrt{2^{n+1}}}\sum_{x\in\{0,1\}^n}(-1)^{f(x)}\ket{x}\ket{1}=:\ket{\Psi_{\widetilde{f}}^\text{ph}},\]
     where $\widetilde{f}(x_1,\ldots, x_n,x_{n+1})=(\mathbf{1}_{x_{n+1}=1}) f(x_1,\ldots, x_n)$.

     So, in fact the function states for functions in $F_0$ and $F_1$ are unitarily equivalent to phase states for another set of twin ensembles $\widetilde{F}^0$ and $\widetilde{F}^1$ with corresponding $\widetilde{M}$ obtained by appending $1$ to every string in $M$ and with $\widetilde{g}$ given by $\widetilde g(x) = (\mathbf{1}_{x_{n+1}=1})\cdot g(x_1,\ldots, x_n)$.
     \Cref{thm:distinguishability-pair-ens} implies these phase state ensembles are distinguishable with constant success probability for $t\geq\Omega(1)$, thus the same holds for the function state ensembles for $F_0$ and $F_1$.

\subsection{Difference matrix: the entries}

Here, to prove \Cref{thm:distinguishability-pair-ens}, we pursue a bound on the trace norm distance in \Cref{eq:tr-dist-f0-f1}.

Define the density matrices
\[A^{(0)} = \mathbb{E}_{B}\left[\left(\ketbra{\Psi_{f_{A,B}^{(0)}}^{\text{ph}}}{\Psi_{f_{A,B}^{(0)}}^{\text{ph}}}\right)^{\otimes t}\right] \qquad \text{and}\qquad A^{(1)} = \mathbb{E}_{B}\left[\left(\ketbra{\Psi_{f_{A,B}^{(1)}}^{\text{ph}}}{\Psi_{f_{A,B}^{(1)}}^{\text{ph}}}\right)^{\otimes t}\right]\,,\]
Call $A:= A^{(0)}-A^{(1)}$ the \emph{difference matrix}.
We now characterize the entries of the difference matrix $A$ by evaluating $A^{(0)}$ and $A^{(1)}$. 

Rows (resp. columns) of $A^{(0)}$ and $A^{(1)}$ are indexed by $t$-tuples $\mb x=(x_1,\ldots, x_t)$ (resp. $\mb y=(y_1,\ldots, y_t)$) of strings $x_j\in \{0,1\}^n$ (resp. $y_j\in \{0,1\}^n$), $1\leq j\leq t$.
It turns out that the entries of $A^{(0)}$ and $A^{(1)}$ depend only on the multiset $\{x_1,\ldots, x_t,y_1,\ldots, y_t\}$, and for this we use the notation $\mb x\cup \mb y$.
In the following it will sometimes be convenient to use $\uparrow$ for exponentiation: $a\uparrow b := a^b$.

\paragraph{Entries of the $A^{(0)}$, the $f^{(0)}$ matrix.}

The $(\mb x, \mb y)$ entry of $A^{(0)}$ takes the form
\[A^{(0)}_{\mb x, \mb y}=\E_{A\sqcup B=M}\left[\frac{1}{2^{nt}}\prod_{z\in \mb x\cup \mb y}(-1)^{f_{A,B}^{(0)}({z})}\right]=\frac{1}{2^{nt}}\E_{A\sqcup B=M}\left[(-1)\uparrow{\sum_{z\in \mb x\cup \mb y}f_{A,B}^{(0)}(z)}\right]\,.\]
Evaluating this expectation gives
\begin{align*}
    A^{(0)}_{\mb x, \mb y}&= \frac{1}{2^{nt}} \left ((-1)\uparrow\sum_{z\in \mb x\cup \mb y, {z}\not\in \cup M} g(z)\right)\prod_{p\in M}\E_{A,B}(-1)\uparrow{\sum_{z\in \mb x\cup \mb y: {z}\in p}f_{A,B}^{(0)}(z)}\\
    &=\frac{1}{2^{nt}} \left((-1)\uparrow\sum_{z\in \mb x\cup \mb y, {z}\not\in \cup M} g(z)\right)\\
    &\hspace{5em}\cdot\prod_{p\in M}\E_{A,B}(-1)\uparrow\begin{cases}
        |\{z\in \mb x\cup \mb y: z\in p\}|\mod 2 &\text{if } p\in A\\
        0 &\text{otherwise}
    \end{cases}\, .
\end{align*}
\normalsize
So, if $|\{z\in \mb x\cup \mb y: z\in p\}|=0$ mod 2, the expectation is always $(-1)^0$.
On the other hand, if $|\{z\in \mb x\cup \mb y:z\in p\}|=1$ mod 2, the expression inside the expectation is $(-1)^1$ w.p. 1/2 and $(-1)^0$ w.p. 1/2.
Define
\begin{equation}
    \label{eq:s-def}
    s(\mb x\cup \mb y)=(-1)\uparrow\sum_{z\in \mb x\cup \mb y, {z}\not\in \cup M} g(z)\,.
\end{equation}
Then
\[A^{(0)}_{\mb x, \mb y}=\begin{cases}
\displaystyle\frac{1}{2^{nt}} s(\mb x\cup \mb y)&\text{if } \forall p\in M, |\{z\in \mb x\cup \mb y: z\in p\}|=0\!\!\mod 2\\[0.8em]
0&\text{otherwise}\,.
\end{cases}\]

\paragraph{Entries of $A^{(1)}$, the $f^{(1)}$ matrix.}

Similarly to above we have 
\begin{align*}
    A^{(1)}_{\mb x, \mb y}&=\frac{1}{2^{nt}} s(\mb x\cup \mb y)\prod_{p\in M}\E_{A,B}(-1)\uparrow{\sum_{z\in \mb x\cup \bm y:x\in p}f_{A,B}^{(1)}(z)}\, .
\end{align*}
To evaluate the expectation, note there are four cases for each $p=(u,v)\in M$, depending on how many times $u$ occurs in $\mb x\cup \mb y$, and how many times $v$ occurs in $\mb x\cup \mb y$.
Denote these quantities mod 2 as $L_p=L_p(\mb x\cup \mb y)$ and $U_p=U_p(\mb x\cup \mb y)$
respectively.
\begin{itemize}
    \item If $L_p=0$ and $U_p=0$, the sum in the exponent is always 0, yielding $(-1)^0=1$ w.p. 1.
    \item If $L_p=0$ and $U_p=1$, the sum is either 0 or 1, each w.p. 1/2, yielding in expectation $(-1)^1/2 + (-1)^0/2=0$.
    \item If $L_p = 1$ and $U_p=0$, we similarly get $0$ for the expectation.
    \item If $L_p=1,U_p=1$, then the sum is always 1, yielding $(-1)^1=-1$ with probability 1.
\end{itemize}
In summary,
\[A^{(1)}_{\mb x, \mb y}=\begin{cases}
\displaystyle\frac{1}{2^{nt}}(-1)^{|\{p\in M:L_p=U_p=1\}|} s(\mb x\cup \mb y)&\text{if } \forall p\in M, |\{z\in \mb x\cup \mb y:z\in p\}|=0\!\!\mod 2\\[0.8em]
0&\text{otherwise}
 \end{cases}\, ,\]
which is equivalent to
\[A^{(1)}_{\mb x, \mb y}=\begin{cases}
 \displaystyle\frac{1}{2^{nt}}(-1)^{|\{p\in M:L_p=U_p=1\}|} s(\mb x\cup \mb y)&\text{if } \forall p\in M, L_p=U_p\\[0.8em]
0&\text{otherwise}
\end{cases}\, .\]
Putting these together, we find
\begin{equation} \label{eq: different_matrix_entries} 
    A_{\mb x, \mb y}=A^{(0)}_{\mb x, \mb y}-A^{(1)}_{\mb x, \mb y}
    =\begin{cases}
        \displaystyle\frac{2}{2^{nt}} s(\mb x\cup \mb y) &\text{if } \forall p\in M, L_p=U_p\\
        & \hspace{2em}\text{ and }|\{p\in M:L_p=U_p=1\}|=1\mod 2\\[0.7em]
         0 &\text{otherwise}
    \end{cases}\, .
\end{equation}

\subsection{Difference matrix: the spectrum}

Here we conduct a fine-grained analysis of the spectrum of the difference matrix $A$ in order to obtain a combinatorial bound on $\|A\|_1$.
In \Cref{sec:analytic-bound-diff-mat} we then understand the asymptotics on this bound in terms of $\epsilon,t,$ and $n$.

Let $D$ be a diagonal matrix with entries $s(\mb x)$ for all $t$-tuples $\mb x$.
Then $\tilde{A}:=2^{nt-1}D^{-1}AD$ is similar to $2^{nt-1}A$ and moreover

\begin{align} \label{equation:entries-for-difference-matrix}
    \tilde{A}_{\mathbf{x}, \mathbf{y}}
    &=\begin{cases}
        s(\mb x)s(\mathbf{x}\cup\mathbf{y})s(\mb y) &\text{if } \forall p\in M, L_p=U_p\text{ and }|\{p\in M:L_p=U_p=1\}|=1\mod 2\\
         0 &\text{otherwise}
    \end{cases}\nonumber\\[0.4em]
    &=\begin{cases}
        1 &\text{if } \forall p\in M, L_p=U_p\text{ and }|\{p\in M:L_p=U_p=1\}|=1\mod 2\\
         0 &\text{otherwise}
    \end{cases}\, ,
\end{align}
Here we used that $s(\mb x \cup \mb y) = s(\mb x)s(\mb y)$; the other notation (for $L_p$ and $U_p$) is as above.

Because of the matrix similarity, we know that $A$ and $\tilde{A}$ have the same spectrum up to a scaling factor of $1/2^{nt-1}$.
It will turn out that after a certain permutation of indices, $\tilde{A}$ is block-diagonal, with each block corresponding to the adjacency matrix of a complete bipartite graph.
Towards defining this block structure, we write $\mathbf{x} C \mathbf{y}$ (read ``$\mathbf{x}$ is compatible with $\mathbf{y}$'') if for all $p\in M, L_p=U_p\text{ and }|\{p\in M:L_p=U_p=1\}|=1\mod 2$.

Given an index tuple $\mb x$, we define some technical quantities of $\mb x$ that are important for combinatorics to follow.
These do not depend on the order of elements in $\mb x$ so we will treat $\mb x$ as a multiset for this discussion.
The multiset $\mb x$ may be partitioned as:
\begin{equation}
\label{eq:partition}
    \mb x = \{\text{ pairs }\} \sqcup \{ \text{ singletons }\} \sqcup \{\text{ elements outside of $\cup M$ }\}
\end{equation}
To form the ``pairs'' multiset, we greedily take as many copies of each pair $p\in M$ as we can from $\mb x$.
The ``singletons'' multiset are the remaining elements in $\mb x$ from $\cup M$ that cannot be paired up, and the final part corresponds to those elements in $\mb x$ outside of $\cup M$.
This partitioning is unique.

For example, suppose $M=\{(1,2),(3,4),(5,6)\}$, where we identify natural numbers with their $n$-bit binary expansions.
Then, the following multiset has the partition
\[ \{0,1,1,1,1,2,2,3,3,3,3,4\} \;\longmapsto\; \{ (1,2),(1,2),(3,4) \}\sqcup\{ 1,1,3,3,3\}\sqcup\{ 0 \}\,.\]

Now define $E(\mb x)$ to be the number of pairs (with multiplicity) in $\mb x$ mod 2; \emph{i.e.,} the cardinality mod 2 of the first part of the partition in \Cref{eq:partition}.
Also, define the \emph{singleton set} of $\mb x$ as
\[\sing(\mb x):=\{e\in \text{ ``singletons'' } : e\text{ occurs an odd number of times in ``singletons'' }\}\,.\]
So, continuing our example, $E(\mb x)=1$ and $\sing(\mb x)=\{3\}$.
We also define the \emph{type} of $\mb x$, $\type(\mb x)$, to be the pairs with nonzero overlap with its singleton set:
\[\mathrm{type}(\mb x)=\{p\in M: p\cap \sing(\mb x)\neq \emptyset\}.\]
Continuing our example, we have that $\type\big((3,5,5,5)\big) = \{(3, 4), (5, 6)\}$. 
We will also need a quantity on pairs of tuples $\mb x,\mb y$ counting the number of elements in their singleton sets that are paired up, mod 2:
\[P(\mb x, \mb y)= |\{p\in M:p\subseteq \sing(\mb x)\cup \sing(\mb y)\}|\mod 2\,.\]
So for example, keeping $M$ as before, $P\big((1,3),(1,4)\big) = 1$.

\begin{lemma}\label{lemma:compatibility-condition-base}
    $\mb x C \mb y$ if and only if
    \[\type(\mb x)=\type(\mb y)\qquad \text{and}\qquad E(\mb x) + E(\mb y) + P(\mb x,\mb y) = 1 \mod 2\,.\]
\end{lemma}
\begin{proof}
    ($\implies$)
    Suppose by way of contradiction that $\mb xC\mb y$ but $\type(\mb x)\neq \type(\mb y)$.
    Then there exists a pair $p=(a,b)$ such that (WLOG) $a$ occurs in the ``singletons'' partition of $\mb x$ an even number of times, and $a$ or $b$ occurs in the ``singletons'' partition of $\mb y$ an odd number of times.
    This implies $L_p+U_p = 1 \mod 2.$

    So now suppose $\type(\mb x)=\type(\mb y)$.
    Because $\mb x C\mb y$, we have $|\{p\in M: L_p=U_p=1\}|= 1\mod 2$.
    But $|\{p\in M: L_p=U_p=1\}|\mod 2$ is precisely $E(\mb x)+E(\mb y)+P(\mb x,\mb y)\mod 2$, because both quantities count the total number of pairs (with multiplicity) occurring in $\mb x\cup \mb y$.
    
    ($\impliedby$)
    If $\type(\mb x) = \type(\mb y)$, we must have that $L_p+U_p = 0 \mod 2$ for all $p$. 
    The elements in the ``pairs'' partition do not affect this condition. 
    For all pairs $p=(a, b)$, either both $a, b$ occur an even number of times in the ``singletons'' partition of $\mb x$ and $\mb y$. Or, if (WLOG) $a$ occurs an odd number of times in $\mb x$, then $a$ or $b$ occurs an odd number of times in $\mb y$, preserving the condition. 
    
    Furthermore, as before, $|\{p\in M: L_p=U_p=1\}|\mod 2$ counts the same quantity as $E(\mb x)+E(\mb y)+P(\mb x,\mb y)\mod 2$.
    Therefore, if $\type(\mb x) = \type(\mb y)$ and $E(\mb x) + E(\mb y) + P(\mb x,\mb y) = 1 \mod 2$, we satisfy the criteria for compatibility.
\end{proof}

To understand the spectrum of the difference matrix we will make repeated use of the following structural fact about compatibility.

\begin{lemma}\label{lemma:compatibility-condition-for-rows}
    Let $\mathbf{x}, \mathbf{x}',\mathbf{x}''$ be such that $\mathbf{x}C\mathbf{x}''$ and $\mathbf{x}'C\mathbf{x}''$. Then for all $\mathbf{y}$, $\mathbf{x}C\mathbf{y}$ if and only if $\mathbf{x}'C\mathbf{y}$. 
\end{lemma}
\begin{proof}
    By symmetry we need only argue the forward direction.
    Clearly $\type(\mb x) = \type(\mb x') = \type(\mb y)$.
    Note that for any $\mb{z}, \mb{z'}, \mb{z''}$, we have $P(\mb{z}, \mb{z'}) + P(\mb{z'}, \mb{z''}) = P(\mb{z}, \mb{z''}) \pmod{2}$, and because $\mb xC \mb x''$ and $\mb x'C \mb x''$,
    \Cref{lemma:compatibility-condition-base} implies
    \[E(\mb x) + P(\mb x, \mb x'') = E(\mb x') + P(\mb x', \mb x'') \pmod 2.\]
    Then we have the following equivalences modulo $2$:
    \begin{align*}
        1&\equiv E(\mb x) + P(\mb x, \mb y) + E(\mb y)\tag{\Cref{lemma:compatibility-condition-base}}\\
        &\equiv E(\mb x) + P(\mb x, \mb x') + P(\mb x'', \mb y) + E(\mb y)\\
        &\equiv E(\mb x') + P(\mb x', \mb x'') + P(\mb x'', \mb y) + E(\mb y)\\
        &\equiv E(\mb x') + P(\mb x', \mb y) + E(\mb y)\,.
    \end{align*}
    Appealing to \Cref{lemma:compatibility-condition-base}, we conclude $\mb x' C \mb y$.
\end{proof}
\noindent This implies, for example, that the $\mb x$\textsuperscript{th} and $\mb x'$\textsuperscript{th} rows in $\tilde{A}$ are equal.

We can view $\tilde{A}$ as describing a graph $G_{\tilde{A}}$ with vertices $(\{0,1\}^n)^t$ and edge set
    \[\{(\mathbf{x}, \mathbf{x}')\in V\times V~|~ \mathbf{x} C \mathbf{x}'\}\]
With \Cref{lemma:compatibility-condition-base,lemma:compatibility-condition-for-rows} in hand, we are prepared to describe the structure of $G_{\tilde{A}}$.
It will be useful to define certain combinatorial quantities first.

\begin{definition}[Combinatorial quantities]
    \label{def:comboquant}
    Let $\Sigma=\Sigma_\text{odd}\sqcup \Sigma_\text{even}\sqcup \Sigma_\text{rest}$ be an alphabet of cardinality $|\Sigma|=2^n$ partitioned such that $\Sigma_\text{odd}$ consists of $p$ pairs, so $|\Sigma_\text{odd}|=2p$, and $\Sigma_\text{even}$ consists of $m-p$ pairs, so $|\Sigma_\text{even}|=2(m-p)$.
    Define $T(t,p)$ as the number of strings of length $t$ over $\Sigma$ such that symbols from $\Sigma_\text{odd}$ each occur an odd number of times, symbols from $\Sigma_\text{even}$ each occur an even number of times, and symbols from $\Sigma_\text{rest}$ occur any number of times.
    Then define
    \begin{equation}
    \label{eq:defnx1x2}
        x_1(t) = \sum_{\substack{p=0\\p \text{ odd}}}^{\min\{m,t\}}\binom{m}{p}T(t,p)\qquad\text{ and}\qquad x_2(t) = \sum_{\substack{p=0\\p \text{ even}}}^{\min\{m,t\}}\binom{m}{p}T(t,p)\,.
    \end{equation}
    Further, define $N(t,p)$ as the number of strings of length $t$ over $\Sigma$ such that for each pair in $\Sigma_\text{odd}$, one of the two symbols occurs an odd number of times while the other occurs an even number of times, and for each pair in $\Sigma_\text{even}$, either both symbols occur an odd number of times or both symbols occur an even number of times.
\end{definition}

\begin{lemma}[Structure of $G_{\tilde{A}}$] \label{lemma:graph-bipartite-component-sizes}
    The graph $G_{\tilde{A}}$ has exactly $\textstyle\sum_{k=0}^{\min(t, m)} \binom{m}{k}$ connected components, each associated with a specific $\type(\cdot)$ of vertex.

    The connected component of $G_{\tilde{A}}$ corresponding to the unique $\type$ of cardinality 0 (the \emph{empty type}) is a complete bipartite graph $(U,V,E)$ with parts $U,V$ such that
    \[|U|=x_1(t)\quad\text{ and }\quad |V|= x_2(t)\,.\]

    For $k\geq 1$, there are $\binom{m}{k}$ connected components, each corresponding to a $\type$ of cardinality $k$.
    All such components are complete bipartite graphs $(U,V,E)$ with parts $U,V$ such that
    \[|U|,|V| = N(t,k)/2\,.\]
\end{lemma}

\begin{proof}
    First, by \Cref{lemma:compatibility-condition-base}, only vertices of the same type can be compatible, and there are a total of $\sum_{k=0}^{\min(t, m)} \binom{m}{k}$ types. 
    For the remainder of the proof, we only need to consider vertices of the same type.
    
    We will analyze the case with $k \geq 1$ first. 
    Here $k$ refers to the size of the type of the vertex. 
    Consider a tuple $\mb{x}$ such that $\sing(\mb{x}) = \{a_1, a_2, \ldots, a_{k-1}, a_k\}$ and $E(\mb{x}) = 0$. 
    Consider a second tuple $\mb{x}'$ such that $\sing(\mb{x}') = \{a_1, a_2, \ldots, a_{k-1}, b_k\}$, $E(\mb{x}') = 0$, and $(a_k, b_k)$ are a pair in $M$.
    Two such tuples must exist since $k > 0$. 
    Note that $\mb x$ is compatible with $\mb x'$. Furthermore, all elements of the same type as $\mb x$ and $\mb x'$ must be compatible with exactly one of $\mb x$ or $\mb x'$. 
    To see that we form a complete bipartite graph, consider two elements $\mb y, \mb y'$ such that $\mb x C \mb y$ and $\mb x' C \mb y'$. 
    By \Cref{lemma:compatibility-condition-for-rows}, we can conclude that $\mb y C \mb y'$. 
    Since $\mb y, \mb y'$ were arbitrary, we have a complete bipartite graph. 
    
    To get the size of each of the two sets in the bipartition, fix a type of size $k\geq 1$ and take $(a_1,b_1),\ldots,(a_k,b_k)$ to be the pairs representing the type. 
    By our definition of type, we know: For each $i$, the elements $a_i$ and $b_i$ appear with differing parities; for each of the remaining pairs in $M$, the two elements of the pair occur with the same parity; and each remaining element, which does not belong to any pair in $M$, can occur with an arbitrary parity.
    Thus, the number of vertices with type $k$ is exactly $N(t,k)$. It remains to observe that the two components in the bipartition for type $k$ are of equal size. To see this, take any tuple $\mb x$ of type $k$ and w.l.o.g.~suppose that $a_1$ occurs with odd parity in $\mathrm{sing}(\mb x)$. (Otherwise, $b_1$ occurs with odd parity in $\mathrm{sing}(\mb x)$ and the remaining argument is easily modified.)
    If we construct a tuple $\mb x'$ from $\mb x$ by replacing as many occurrences of $a_1$ as are in $\mathrm{sing}(\mb x)$ with $b_1$ (while keeping the remaining elements the same), then $\mb{x}C\mb{x}'$ by \Cref{lemma:compatibility-condition-base}. This provides us with a one-to-one mapping between the two components in the bipartition, thus they are of equal size.
    
    For $k = 0$, we instead consider two distinct vertices $\mb x$ and $\mb x'$ with $E(\mb x) = 0$ and $E(\mb x') = 1$ and $|\sing (\mb x)| = |\sing (\mb x')| = 0$. 
    Note that $\mb x$ and $\mb x'$ are compatible. 
    By an analogous argument to the case of $k > 0$, we must have all vertices connected to either $\mb x$ and $\mb x'$, and we get a complete bipartite graph. 
    
    For the sizes of the sets in the bipartition, note that the number of vertices connected to $\mb x$ with $E(\mb x) = 0$ will be precisely $x_1(t)$ and the number of vertices connected to $\mb x'$ with $E(\mb x') = 1$ will be precisely $x_2(t)$, completing our proof. 
\end{proof}

As a consequence of this lemma, when suitably ordering the indexing tuples, the matrix $\tilde{A}$ has a block-diagonal structure with blocks corresponding to the connected components of $G_{\tilde A}$.
Thus, to determine the spectrum of $\tilde{A}$, it suffices to determine the spectrum of each block.

The adjacency matrix of a complete bipartite graph between $a$- and $b$-many vertices has two nonzero eigenvalues, each of magnitude $\sqrt{ab}$ \cite[Chapter VIII.2]{Bollobs1998}.
Instantiating this fact in the context of \Cref{lemma:graph-bipartite-component-sizes} and renormalizing by $1/2^{nt-1}$ (recall $A=(D\tilde{A} D^{-1})/2^{nt - 1}$), we obtain the following bound on $\|A\|_1$.

\begin{corollary}
\label{corollary:combinatorial-bound-on-A}
    The $1$-norm of the matrix $A$ satisfies:
    \begin{equation*}
    \|A\|_1= \frac{2}{2^{nt - 1}}\sqrt{x_1(t)\cdot x_2(t)} + \frac{2}{2^{nt - 1}}\sum_{k=1}^{\min\{t,m\}} \binom{m}{k} \frac{N(t,k)}{2}.
    \end{equation*}
\end{corollary}

\subsection{Difference matrix: the trace norm}
\label{sec:analytic-bound-diff-mat}
Here we analyze the growth of $\|A\|_1$ in terms of $\epsilon,t,$ and $n$.
We begin by using exponential generating functions (or EGFs---see \cite{WILF199430} for background) to derive explicit expressions for $T(t,p)$, $N(t,p)$, and $x_1(t)+x_2(t)$.

\begin{lemma}
    Let $T(t,p)$ be as in \Cref{def:comboquant}.
    Then
	\[T(t,p)=\left(\frac12\right)^{2m}\sum_{j=0}^{2m-2p}\sum_{k=0}^{2p}(-1)^k\binom{2m-2p}{j}\binom{2p}{k}(2^n-2j-2k)^t.\]
\end{lemma}
\begin{proof}
	For readability let us use $a$ for the number of symbols occurring an even number of times, $b$ for the number of symbols occurring an odd number of times and $c$ for the rest.
    Then $T(t,p)$ has EGF
    \[f(x)=\left(\frac{e^x+e^{-x}}{2}\right)^a\left(\frac{e^x-e^{-x}}{2}\right)^b\big(e^x\big)^c\]
    Rearranging, we find
    \[f(x)=\left(\frac12\right)^{a+b}\sum_{j=0}^a\sum_{k=0}^b\binom{a}{j}\binom{b}{k}(-1)^ke^{(a-2j+b-2k+c)x}\,.\]
    And we assume $c > a+b$, so we use that the EGF $e^{\eta x}$ corresponds to the sequence $\{\eta^t\}_{t=0}^\infty$, read off the relevant coefficient to derive the formula for $T(t,p)$:
    \[T(t,p)=\left(\frac12\right)^{a+b}\sum_{j=0}^a\sum_{k=0}^b\binom{a}{j}\binom{b}{k}(-1)^k(a-2j+b-2k+c)^t\]
    Substituting for $a,b,c$ yields the result.
\end{proof}

\begin{lemma}
With $x_1$ and $x_2$ defined as in \Cref{eq:defnx1x2},
\label{lem:x1plusx2}
	\[x_1(t)+x_2(t)=\left(\frac12\right)^m\sum_{k=0}^m\binom{m}{k}(2^n-4k)^t.\]
\end{lemma}
\begin{proof}
The combinatorial interpretation of $x_1(t)+x_2(t)$ is the number of strings of length $t$ over $\Sigma$ where symbols from $A\sqcup B, |A\sqcup B|=2m$ are paired up and within each pair, they must appear with the same parity.

The EGF for strings with 2 elements appearing with same parity is
\[\left(\frac{e^x+e^{-x}}{2}\right)^2 + \left(\frac{e^x-e^{-x}}{2}\right)^2 = \frac{e^{2x}+e^{-2x}}{2}\,.\]
To construct the desired strings, we combine $m$ copies of this with $2^n-m$ copies of the unrestricted EGF, leading to the overall EGF
\[2^{-m}\left(e^{2x}+e^{-2x}\right)^me^{(2^n-2m)x}\, .\]
The result follows from simplifying this EGF and recognizing the related counting formula.
\end{proof}

\begin{lemma}\label{lem:series-N}
    Let $N(t,p)$ be as in \Cref{def:comboquant}.
    Then
	\[N(t,p)=\left(\frac12\right)^{m}\sum_{j=0}^{p}\sum_{k=0}^{m-p} (-1)^j \binom{p}{j}\binom{m-p}{k}(2^n-4j-4k)^t.\]
\end{lemma}
\begin{proof}
    The proof uses EGFs analogously to how we derived the expression for $T(t,p)$. 
    The EGF for $N(t,p)$ is
    \begin{equation*}
        f(x)
        = \left(\frac{e^{2x}-e^{-2x}}{2}\right)^{p}\left(\frac{e^{2x}+e^{-2x}}{2}\right)^{m-p}\big(e^x\big)^{2^n-2m} \, .
    \end{equation*}
    We can now rearrange this product of sums into a sum of products and, using again that the EGF $e^{\eta x}$ corresponds to the sequence $\{\eta^t\}_{t=0}^\infty$, read off the expression for $N(t,p)$.
\end{proof}

\begin{lemma}\label{lem:relation-sum-N-x1-x2}
    With $x_1$ and $x_2$ defined as in \Cref{eq:defnx1x2},
    \begin{equation*}
        \frac{1}{2^{nt}}\sum_{k=1}^{\min\{t,m\}} \binom{m}{k} \frac{N(t,k)}{2}
        = \frac{1}{2}\left( 1 -  \frac{x_1(t)+x_2(t)}{2^{nt}}\right)\, .
    \end{equation*}
\end{lemma}
\begin{proof}
    First, observe that $\sum_{k=0}^{\min\{t,m\}} \binom{m}{k} N(t,k)=2^{nt}$, since every one of the overall $2^{nt}$ tuples belongs to some type and we are summing over all sizes of types. Consequently, 
    \begin{align*}
        \frac{1}{2^{nt}}\sum_{k=1}^{\min\{t,m\}} \binom{m}{k} \frac{N(t,k)}{2}
        &= \frac{1}{2} - \frac{1}{2}\binom{m}{0}\frac{N(t,0)}{2^{nt}}\\
        &= \frac{1}{2} - \frac{1}{2^{m+1}}\sum_{k=0}^m \binom{m}{k} \left(1-\frac{4k}{2^n}\right)^t\\
        &= \frac{1}{2}- \frac{1}{2} \frac{x_1(t)+x_2(t)}{2^{nt}}\, ,
    \end{align*}
    where the second-to-last step used \Cref{lem:series-N} and the last step used \Cref{lem:x1plusx2}.
\end{proof}

We now further study the trace norm of $A$.
We will need the following technical fact.

\begin{proposition}
    \label{prop:concentration}
    Suppose $m=|M|=\epsilon2^n$ and $\epsilon=\epsilon(n)<1$.
    Let $D\in\mathbb{R}$ be a fixed constant.
    Then, for any $t=t(n)\leq \mathcal{O}(2^{bn})$ with $0\leq b < \frac{1}{4}$, we have
    \begin{equation*}
        \E_{K\sim\mathcal{B}(m, 1/2)}\left(1 - \frac{D K }{2^n}\right)^t
        = \left(1 - \frac{D \frac{m}{2}}{2^n}\right)^t + o(1)\, .
    \end{equation*}
    
\end{proposition}
\noindent In \Cref{prop:concentration} and what follows, $\mc B(\ell,1/2)$ denotes the Binomial distribution with $\ell$ trials and success probability $1/2$.

\begin{proof}
    Notice that, for any $t\in\mathbb{N}$, the function $f:[0,B)\to\mathbb{R}$ given by $f(x) = (1 - x)^t$ is Lipschitz with Lipschitz constant $t\max_{0\leq x\leq B}|1-x|^{t-1}$. Using that $|1 - Dk/2^n|\leq 1$ holds for all $0\leq k\leq m$ for sufficiently large $n$, this implies 
    \begin{align*}
        \E_{K\sim\mathcal{B}(m, 1/2)}\left\lvert\left(1 - \frac{D K }{2^n}\right)^t
        - \left(1 - \frac{D \frac{m}{2}}{2^n}\right)^t\right\rvert
        &\leq \E_{K\sim\mathcal{B}(m, 1/2)}\left\lvert t \left(\frac{DK}{2^n} - \frac{D\frac{m}{2}}{2^n}\right) \right\rvert\, .
    \end{align*}
    From Chernoff we have that
    \[\Pr\left[\left|\frac{DK}{2^n}-\frac{D\frac{m}{2}}{2^n}\right|\geq \eta \right]\leq 2\exp\left(-\mathsf{const}\cdot2^n\cdot\frac{\eta^2}{\epsilon}\right).\]
    Call the low-probability event above $E$.
    Then
    \begin{align*}
        \E_{K\sim\mathcal{B}(m, 1/2)}\left\lvert t \left(\frac{DK}{2^n} - \frac{D\frac{m}{2}}{2^n}\right) \right\rvert
        &\leq t\Pr[E^c]\eta + t\Pr[E]\frac{D\frac{m}{2}}{2^n}\\
        &\leq t\eta + 2tD\varepsilon\exp\left(-\mathsf{const}\cdot2^n\cdot\frac{\eta^2}{\epsilon}\right)\, .
    \end{align*}
    We can set $\eta=\min\{1/n, 1/t^2\}$, then, because of our assumption on $t=t(n)$, both summands go to $0$ as $n\to\infty$, finishing the proof.
\end{proof}

\begin{theorem}
    $\|A\|_1= \Omega(1)$ for $t=\Omega\big(\epsilon(n)^{-1}\big)$, assuming $\varepsilon=\varepsilon(n)\geq\Omega(2^{-bn})$ with $0\leq b<\frac{1}{4}$.
\end{theorem}
\begin{proof}

Let us label the expression for $\|A\|_1$ from Corollary \ref{corollary:combinatorial-bound-on-A} as
\begin{equation}
\label{eq:3-star-bound}
    \|A\|_1
    = \vphantom{\sum_{k=1}^{\min\{t,m\}}}\frac{2}{2^{nt-1}}\sqrt{x_1(t)\cdot x_2(t)} + \underbrace{\frac{2}{2^{nt-1}}\sum_{k=1}^{\min\{t,m\}} \binom{m}{k} \frac{N(t,k)}{2}}_{(*)}.
\end{equation}
We will lower-bound $\|A\|_1$ by lower-bounding the second summand here, $(*)$.
Lemma \ref{lem:x1plusx2} implies
\begin{equation}
\label{eq:x1plusx2-prob}
    x_1(t)+x_2(t)=\left(\frac12\right)^m\sum_{k=0}^m\binom{m}{k}\left(2^n-4k\right)^t=\E_{K\sim \mc B(m,1/2)}\left(2^n-4K\right)^t
\end{equation}
where $\mc B(m,1/2)$ refers to the Binomial distribution.
Returning to $(*)$ and using \Cref{lem:relation-sum-N-x1-x2}, we have
\begin{align*}
    (*) 
    &= \frac{4}{2}\left(1 - \frac{x_1(t) + x_2(t)}{2^{nt}}\right)
    = 2\left(1 - \mathbb{E}\left(1 - \frac{4K}{2^n}\right)^t\right) \, .
\end{align*}
Using \Cref{prop:concentration}, we obtain that, as long as $t\leq \mathcal{O}(2^{bn})$ with $0\leq b<\frac{1}{4}$,
\begin{equation*}
    (*) 
    = 2\left(1 - \left(1 - \frac{2m}{2^n}\right)^t + o(1)\right)
    = 2\left(1 - \left(1 - 2\varepsilon\right)^t + o(1)\right)\, .
\end{equation*}
Notice that our assumption on $\varepsilon=\varepsilon(n)$ ensures that $\frac{1}{\varepsilon(n)}\leq \mathcal{O}(2^{bn})$ with $0\leq b<\frac{1}{4}$. Therefore, we can consider $t\geq \Omega(1/\varepsilon)$, we get $1 - \left(1 - 2\varepsilon\right)^t\geq \Omega(1)$, and therefore $\|A\|_1\geq (*)\geq \Omega(1)$.
\end{proof}

\begin{remark}
We can extend the above quantum distinguishability analysis to the ensembles from \cite{black_nearly_2023}. 
The construction in \cite{black_nearly_2023}, based on \textit{Talagrand’s random DNFs} \cite{talagrand1996much}, establishes a lower bound of $\exp(\Omega(\sqrt{n}/\varepsilon))$ for passive classical monotonicity testing via a birthday paradox argument.
The construction randomly selects DNF terms of fixed width to define a partial partition of the Boolean cube into disjoint sets $U_j$ such that any two points in different $U_j$ are incomparable. 
The difference between the monotone $D_{\text{yes}}$ and non-monotone $D_{\text{no}}$ case lies in the function value assignments: in $D_{\text{yes}}$, values within each disjoint set $U_j$ are structured monotonically while in $D_{\text{no}}$, values with each $U_j$ are randomly assigned. Classically, distinguishing these distributions requires $\exp(\Omega(\sqrt{2^n}/\epsilon))$ samples, as a tester must sample at least two points from the same $U_j$ to gain information. 
This leads to an exponential lower bound when parameters are chosen appropriately.

Following an argument structured similarly to the one above, one may see that the difference matrix between the induced function state ensembles in the quantum setting decomposes into blocks corresponding to complete multipartite graphs. 
To see this, in analogy to the analysis from \Cref{section:passive-quantum-monotonicity-testing-lower-bound}, we can define a notion of compatibility between any two index tuples. 
Given a collection of sets $U_j$ and an index tuple $\mb x$, we first remove duplicates from the tuple and analyze its intersection pattern with each $U_j$. 
For instance, if $U_1 = \{1, 2\}$, $U_2 = \{3, 4\}$, and $U_3 = \{5, 6\}$, and our index tuple is $(1, 1, 1, 2, 3, 3, 3, 3, 5)$, then after removing duplicate, the corresponding subsets under the $U_j$ sets are $[{1, 2}], [\ ], [{5}]$.
Two tuples are said to be compatible if, for every $j$, the number of elements from each $U_j$ that appear in the tuple is even but not identical across the tuples after removing duplicates and decomposing. 
For example, the tuple $(1, 1, 1, 2, 3, 3, 3, 3, 5)$ is compatible with $(1, 1, 1, 2, 3, 3, 3, 3, 6)$ but not with $(1, 2, 3, 3, 3, 3, 3, 3, 5)$, as the latter shares the same decomposition as the original.

Importantly, this compatibility is transitive: if $\mb x$ is compatible with $\mb x'$ and if $\mb x'$ is compatible with $\mb x''$, then $\mb x$ is also compatible with $\mb x''$. 
This transitivity induces complete multipartite graph blocks with each block corresponding to a compatibility class.
Therefore, the trace distance between the two ensembles equals the sum of the trace norms of multipartite graphs of various sizes.
Bounds on the eigenvalues of such a graph in terms of the sizes of its parts can be found in, for example, \cite{MR593991, mehatari2023eigenvalues}.
Also in analogy to our analysis of the \cite{Goldreich2000} construction, the construction from \cite{black_nearly_2023} results in a graph such that sum of the contributions from the block components yields distinguishability in the quantum case. 
The ensembles remain distinguishable when $t = \Omega(1/\epsilon)$, and thus they achieve no improvement over the generic lower bound. 

\end{remark}

\section*{Acknowledgments}
The authors thank Srinivasan Arunachalam, Fernando Jeronimo, Kyle Gulshen, Jiaqing Jiang, John Bostanci, Yeongwoo Hwang, Shivam Nadimpali, John Preskill, Akshar Ramkumar, Abdulrahman Sahmoud, Mehdi Soleimanifar, and Thomas Vidick for enlightening discussions.
Moreover, we thank Nathan Harms for suggesting we look at monotonicity and general discussions on classical bounds for related problems.
We thank Fermi Ma for helpful discussions that provided inspiration for \Cref{thm:qsamp-cquery-sep}.
We are grateful to Francisco Escudero Gutiérrez for pointing us to the Fourier-analytic characterization of monotonicity which allowed us to improve a previous passive quantum monotonicity tester.
Finally, we thank the anonymous STOC 2025 reviewers for helpful feedback.
MCC was partially supported by a DAAD PRIME fellowship and by the BMBF (QPIC-1). 
JS is funded by Chris Umans' Simons Foundation Investigator Grant.
Part of this work was completed while JS was visiting the Simons Institute for the Theory of Computing, supported by DOE QSA grant \#FP00010905.

\appendix

\section{Passive quantum testers for symmetry and triangle-freeness} \label{sec:passive-quantum-testers}

\subsection{Passive quantum symmetry testing}\label{subsection;passive-quantum-symmetry-testing}

A function $f:\{0,1\}^n\to\{0,1\}$ is called symmetric if $f\circ \pi = f$ holds for all permutations $\pi \in S_n$. Here, $\pi\in S_n$ acts on $n$-bit strings by permuting coordinates. That is, $\pi(x_1\ldots x_n)=x_{\pi^{-1}(1)}\ldots x_{\pi^{-1}(n)}$.
This gives rise to the following classical testing problem:

\begin{problem}[Classical symmetry testing]\label{problem:classical-symmetry-testing}
  Given query access to an unknown function $f:\{0,1\}^n\to\{0,1\}$ and an accuracy parameter $\varepsilon\in (0,1)$, decide with success probability $\geq 2/3$ whether
  \begin{itemize}
      \item[(i)] $f$ is symmetric, or
      \item[(ii)] $f$ is $\varepsilon$-far from all symmetric functions, that is, we have $\Pr_{x\sim\{0,1\}^n}[f(x)\neq g(x)]\geq\varepsilon$ for all symmetric functions $g:\{0,1\}^n\to\{0,1\}$,
  \end{itemize}  
  promised that $f$ satisfies either (i) or (ii).  
\end{problem}

Symmetry allows for a (trivial) reformulation in terms of (in general non-local) pairwise comparisons:

\begin{proposition}
    A function $f:\{0,1\}^n\to\{0,1\}$ is symmetric if and only if for all $x\in\{0,1\}^n$ and for all $\pi\in S_n$, the equality
    \begin{equation}
        f(x)=f(\pi(x))
    \end{equation}
    holds.
\end{proposition}

This characterization becomes important for testing because of the following result:

\begin{theorem}[Soundness of symmetry testing (compare {\cite[Lemma 3.3]{blais2015partially}})]\label{theorem:soundness-classical-symmetry-testing}
    If $f:\{0,1\}^n\to\{0,1\}$ is exactly $\varepsilon$-far from all symmetric functions, then
    \begin{equation}\label{eq:oundness-classical-symmetry-testing}
        \varepsilon
        \leq \Pr_{x\sim\{0,1\}^n, \pi\sim S_n}[f(x)\neq f(\pi(x))]
        \leq 2\varepsilon\, .
    \end{equation}
\end{theorem}

\Cref{theorem:soundness-classical-symmetry-testing} implies that we can classically test symmetry from query access simply by sampling a random permutation $\pi$ and a random input $x$ and then comparing the function values $f(x)$ and $f(\pi(x))$. Here, $\sim 1/\varepsilon$ many queries suffice to achieve success probability $\geq 2/3$ in symmetry testing.

We now describe how to make use of \Cref{theorem:soundness-classical-symmetry-testing} to build a passive quantum symmetry tester.

\begin{theorem}[Passive quantum symmetry testing]\label{theorem:quantum-symmetry-testing}
    There is an efficient quantum algorithm that uses $\mathcal{O}\left(\frac{\log(1/\delta)}{\varepsilon^2}\right)$ many copies of the function state $\ket{\Psi}=\frac{1}{\sqrt{2^n}}\sum_{x\in\{0,1\}^n}\ket{x,f(x)}$ to decide, with success probability $\geq 1-\delta$, whether $f$ is symmetric or $\varepsilon$-far from all symmetric functions.
\end{theorem}
\begin{proof}   
    For a permutation $\pi\in S_n$, write $P(\pi)$ for the representation of that permutation on $(\mathbb{C}^2)^{\otimes n}$ given as $P(\pi) = \sum_{x\in\{0,1\}^n} \ket{\pi(x)}\bra{x}$. Then, the orthogonal projector onto the symmetric subspace of $(\mathbb{C}^2)^{\otimes n}$ can be written as
    \begin{equation*}
        P_{\mathrm{sym}}^{n}
        = \frac{1}{n!} \sum_{\pi\in S_n} P(\pi)\, .
    \end{equation*}
    Notice that, if $\ket{\Psi}$ is the function state for $f:\{0,1\}^n\to\{0,1\}$, then
    \begin{equation*}
        \bra{\Psi}(P_{\mathrm{sym}}^{n}\otimes \mathds{1}_2)\ket{\Psi}
        = \Pr_{x\sim\{0,1\}^n, \pi\sim S_n}[f(x)\neq f(\pi(x))]\, .
    \end{equation*}
    So, by a Chernoff-Hoeffding bound, we can, with success probability $\geq 1-\delta$, obtain a $(\varepsilon/3)$-accurate estimate of the probability in \Cref{eq:oundness-classical-symmetry-testing} by independently performing the two-outcome projective measurement $\{P_{\mathrm{sym}}^{n}\otimes \mathds{1}_2, \mathds{1}_2^{\otimes (n+1)}-P_{\mathrm{sym}}^{n}\otimes \mathds{1}_2\}$ on $m=\mathcal{O}(\log(1/\delta)/\varepsilon^2)$ many single copies of $\ket{\Psi}$ and then taking the empirical average of the observed outcomes (with outcome $1$ associated to $P_{\mathrm{sym}}^{n}\otimes \mathds{1}_2$).
    As the two-outcome measurement $\{P_{\mathrm{sym}}^{n}\otimes \mathds{1}_2, \mathds{1}_2^{\otimes (n+1)}-P_{\mathrm{sym}}^{n}\otimes \mathds{1}_2\}$ can be implemented efficiently using $\mathcal{O}(n^2)$ auxiliary qubits and $\mathcal{O}(n^2)$ controlled-SWAP gates \cite{barenco1997stabilization, laborde2022quantum}, our quantum symmetry tester is also computationally efficient.
\end{proof}

In contrast to the classical sample complexity of $\Theta(n^{1/4})$ for classical passive symmetry testing \cite{alon2016active}, our passive quantum symmetry tester in \Cref{theorem:quantum-symmetry-testing} achieves an $n$-independent quantum sample complexity. Thus, we have an unbounded separation between classical and quantum for this passive testing task.

Finally, let us comment on two extensions.
Firstly, relying on the second inequality in \Cref{theorem:soundness-classical-symmetry-testing}, we can modify the proof of \Cref{theorem:quantum-symmetry-testing} to obtain an efficient tolerant quantum passive symmetry tester that uses $\mathcal{O}\left(\frac{\log(1/\delta)}{(\varepsilon_2-\varepsilon_1)^2}\right)$ copies of the unknown function state to decide whether $f$ is $\varepsilon_1$-close to or $\varepsilon_2$-far from symmetric, assuming that $\varepsilon_2 > 2C\varepsilon_1$ holds with $C>1$ some constant.
Secondly, as \Cref{theorem:soundness-classical-symmetry-testing} can be extended to so-called partial symmetric functions (compare again \cite[Lemma 3.3]{blais2015partially}), also our passive quantum symmetry tester can be modified to test for partial symmetry.

\subsection{Passive quantum triangle-freeness testing}\label{subsection;passive-quantum-trianglefreeness-testing}

For $x,y\in\{0,1\}^n$ and for a Boolean function $f:\{0,1\}^n\to\{0,1\}$, we say that $(x,y,x+y)$ is a triangle in $f$ if $f(x)=f(y)=f(x+y)=1$. Accordingly, we call the function $f$ triangle-free if no triple $(x,y,x+y)$ is a triangle in $f$.
Testing for triangle-freeness thus becomes the following problem:

\begin{problem}[Classical triangle-freeness testing]\label{problem:triangle-freeness}
  Given query access to an unknown function $f:\{0,1\}^n\to\{0,1\}$ and an accuracy parameter $\varepsilon\in (0,1)$, decide with success probability $\geq 2/3$ whether
  \begin{itemize}
      \item[(i)] $f$ is triangle-free, or
      \item[(ii)] $f$ is $\varepsilon$-far from all triangle-free functions, that is, we have $\Pr_{x\sim\{0,1\}^n}[f(x)\neq g(x)]\geq\varepsilon$ for all triangle-free functions $g:\{0,1\}^n\to\{0,1\}$,
  \end{itemize}  
  promised that $f$ satisfies either (i) or (ii).  
\end{problem}

The natural approach towards testing for triangle-freeness from query access is to choose $x,y\in\{0,1\}^n$ at random and check whether $(x,y,x+y)$ is a triangle in $f$, and to repeat this sufficiently often.
Bounding the number of repetitions needed to succeed with this approach is non-trivial, connecting to Szemerédi's regularity lemma \cite{szemeredi1976regular} and the triangle removal lemma \cite{ruzsa1978triple}.
In our next result, we recall the to our knowledge best known corresponding bounds.

\begin{theorem}[Soundness of triangle-freeness testing \cite{fox2011new, hatami2016arithmetic}]\label{theorem:classical-soundness-triangle-freeness-testing}
    If $f:\{0,1\}^n\to\{0,1\}$ is $\varepsilon$-far from all triangle-free functions, then
    \begin{equation}
        \Pr_{x,y\sim\{0,1\}^n}\left[ f(x)=f(y)=f(x+y)=1 \right]
        \geq \frac{1}{\mathsf{Tower}\left(C\cdot \left\lceil \log\left(\frac{1}{\varepsilon}\right)\right\rceil\right)}\, ,
    \end{equation}
    where $C>0$ is a universal integer constant.
\end{theorem}

Here, $\mathsf{Tower}(i)$ denotes a tower of $2$'s of height $i$. That is, we define the tower function $\mathsf{Tower}:\mathbb{N}\to\mathbb{N}$ inductively via $\mathsf{Tower}(0)=1$ and $\mathsf{Tower}(i+1)=2^{\mathsf{Tower}(i)}$.
In a way familiar by now from the two previous subsections, \Cref{theorem:classical-soundness-triangle-freeness-testing} can be used to show that $\sim \mathsf{Tower}\left(C\cdot \left\lceil \log\left(\frac{1}{\varepsilon}\right)\right\rceil\right)$ many queries suffice for the the simple query-based triangle-freeness tester mentioned above to achieve success probability $2/3$.

We now use \Cref{theorem:classical-soundness-triangle-freeness-testing} to develop a passive quantum triangle-freeness tester.

\begin{theorem}[Passive quantum triangle-freeness testing]\label{theorem:quantum-triangle-freeness-testing}
    There is an efficient quantum algorithm that uses $\tilde{\mathcal{O}}\left(\ln(1/\delta)\left(\mathsf{Tower}\left(C\cdot \left\lceil \log\left(\frac{1}{\varepsilon}\right)\right\rceil\right)\right)^{6}\right)$ many copies of the function state $\ket{\Psi}=\frac{1}{\sqrt{2^n}}\sum_{x\in\{0,1\}^n}\ket{x,f(x)}$ to decide, with success probability $\geq 1-\delta$, whether $f$ is triangle-free or $\varepsilon$-far from all triangle-free functions.
\end{theorem}
\begin{proof}
    Our passive quantum triangle-freeness tester first sets confidence and accuracy parameters $\tilde{\delta}=\delta/(5m)$ and $\tilde{\varepsilon}=\left(\mathsf{Tower}\left(C\cdot \left\lceil \log\left(\frac{1}{\varepsilon}\right)\right\rceil\right)\right)^{-1}$, respectively.
    Then, it repeats the following for $1\leq i\leq m=\left\lceil \frac{18 \ln(10/\delta)}{ \tilde{\varepsilon}^2}\right\rceil$:
    \begin{enumerate}
        \item Take $\frac{\ln(m/\tilde{\delta})}{\varepsilon}$ many copies of $\ket{\Psi}$ and, for each of them, measure the last qubit in the computational basis. If none of these measurements produces outcome $1$, abort this iteration, set $\hat{\mu}_i=0$, and go to the next iteration. Otherwise, take any one of the post-measurement states for which $1$ was observed, measure the first $n$ qubits, let the outcome be $y_i$.
        \item Run the procedure from \Cref{lemma:function-value-subset-state-preparation} on $2\left\lceil 162 \ln(6/\tilde{\delta}) \left(6/\tilde{\varepsilon}\right)^{4}\right\rceil\cdot\left\lceil\frac{\ln(2\left\lceil 162 \ln(6/\tilde{\delta}) \left(6/\tilde{\varepsilon}\right)^{4}\right\rceil / \tilde{\delta})}{\varepsilon}\right\rceil$ many copies of $\ket{\Psi}$ to obtain $2\left\lceil 162 \ln(6/\tilde{\delta}) \left(6/\tilde{\varepsilon}\right)^{4}\right\rceil$ many copies of the post-measurement state $$\ket{\Psi_1} = (|\{x\in\{0,1\}^n: f(x)=1\}|)^{-1/2}\sum_{x\in\{0,1\}^n: f(x)=1} \ket{x}$$ where we threw away the last qubit.
        If the procedure from \Cref{lemma:function-value-subset-state-preparation} outputs FAIL, abort this iteration, set $\hat{\mu}_i=0$, and go to the next iteration. 
        \item Consider the $n$-qubit unitary $U_{y_i}$ acting as $U_{y_i}\ket{x}=\ket{x+y_i}$. Run the procedure from \Cref{lemma:function-value-subset-state-preparation} on $2\left\lceil 162 \ln(6/\tilde{\delta}) \left(6/\tilde{\varepsilon}\right)^{4}\right\rceil\cdot\left\lceil\frac{\ln(2\left\lceil 162 \ln(6/\tilde{\delta}) \left(6/\tilde{\varepsilon}\right)^{4}\right\rceil / \tilde{\delta})}{\varepsilon}\right\rceil$ many copies of $(U_{y_i}\otimes\mathds{1}_2)\ket{\Psi}$ to obtain $2\left\lceil 162 \ln(6/\tilde{\delta}) \left(6/\tilde{\varepsilon}\right)^{4}\right\rceil$ many copies of the post-measurement state $$\ket{\Psi_{y_i,1}} = (|\{x\in\{0,1\}^n: f(x+y_i)=1\}|)^{-1/2}\sum_{x\in\{0,1\}^n: f(x+y_i)=1} \ket{x}\, ,$$ where we threw away the last qubit.
        If the procedure from \Cref{lemma:function-value-subset-state-preparation} outputs FAIL, abort this iteration, set $\hat{\mu}_i=0$, and go to the next iteration. 
        \item Run the procedure from \Cref{corollary:two-wise-intersection-estimation} on $\left\lceil 162 \ln(6/\tilde{\delta}) \left(6/\tilde{\varepsilon}\right)^{4}\right\rceil$ copies of each of $\ket{\Psi}$, $(U_{y_i}\otimes\mathds{1}_2)\ket{\Psi}$ and on $2\left\lceil 162 \ln(6/\tilde{\delta}) \left(6/\tilde{\varepsilon}\right)^{4}\right\rceil$ copies of each of $\ket{\Psi_1}$, $\ket{\Psi_{y_i,1}}$ to produce an $\left(\tilde{\varepsilon}/6\right)$-accurate estimate $\hat{\mu}_i$ of the probability $\Pr_{x\sim\{0,1\}^n}[f(x)=1=f(x+y_i)]$. 
    \end{enumerate}
    Finally, the tester makes a decision as follows: If $\frac{1}{m}\sum_{i=1}^m \hat{\mu}_i \leq \tilde{\varepsilon}/3$, output ``triangle-free''. Otherwise, output ``$\varepsilon$-far from triangle-free''. 

    Let us analyze the completeness and soundness of this tester.
    First, we consider completeness.
    So, assume that $f$ is triangle-free.
    Note the first step above failing in any iteration only ever decreases the empirical average evaluated by the tester in the end, and thus cannot increase the probability of falsely rejecting a triangle-free function. Thus, we can condition on the first step succeeding in all $m$ iterations. In particular, we can assume that $\Pr_{y\sim\{0,1\}^n}[f(y)=1]>0$.
    By a similar argument, we can also condition on the second and third step succeeding in all iterations.
    And given these successes, the fourth step will, by \Cref{corollary:two-wise-intersection-estimation}, produce, with probability $\geq 1-\tilde{\delta}$, a $\hat{\mu}_i$ satisfying $\lvert \hat{\mu}_i - \Pr_{x\sim\{0,1\}^n}[f(x)=1=f(x+y_i)]\rvert\leq \tilde{\varepsilon}/6$.
    By a union bound, this means that, with probability $\geq 1-\delta/5$, the empirical average $\frac{1}{m}\sum_{i=1}^m \hat{\mu}_i$ is a $(\tilde{\varepsilon}/6)$-accurate estimate of $\frac{1}{m}\sum_{i=1}^m \Pr_{x\sim\{0,1\}^n}[f(x)=1=f(x+y_i)]$.
    We can consider the $\Pr_{x\sim\{0,1\}^n}[f(x)=1=f(x+y_i)]$ for $1\leq i\leq m$ as i.i.d.~random variables taking values in $[0,1]$. Hence, by a Chernoff-Hoeffding concentration bound and our choice of $m$, with probability $\geq 1-\delta/5$, we have
    \begin{align*}
        \left\lvert \frac{1}{m}\sum_{i=1}^m\Pr_{x\sim\{0,1\}^n}[f(x)=1=f(x+y_i)] - \mathbb{E}_{y\sim\{0,1\}^n: f(y)=1}\left[\Pr_{x\sim\{0,1\}^n}[f(x)=1=f(x+y)]\right]\right\rvert
        &\leq \tilde{\varepsilon}/6\, .
    \end{align*}
    Noticing that 
    \small
    \begin{align}\label{eq:tirangle-freeness-proof-intermediate}
        \mathbb{E}_{y\sim\{0,1\}^n: f(y)=1}\left[\Pr_{x\sim\{0,1\}^n}[f(x)=1=f(x+y)]\right]
        = \frac{\Pr_{x,y\sim\{0,1\}^n}[f(x)=f(y)=f(x+y)=1]}{\Pr_{y\sim\{0,1\}^n}[f(y)=1]}
        = 0 \, ,
    \end{align}
    \normalsize
    since $f$ was assumed to be triangle-free, we conclude (after one more union bound) that $\frac{1}{m}\sum_{i=1}^m \hat{\mu}_i \leq \tilde{\varepsilon}/3$ holds with probability $\geq 1 - 2\delta/5\geq 1-\delta$, and in this case the tester outputs ``triangle-free'', thus proving completeness.    

    Next, we consider soundness.
    So, assume that $f$ is $\varepsilon$-far from triangle-free.
    This in particular implies that $\Pr_{x\sim\{0,1\}^n}[f(x)=1]\geq \varepsilon$, since the zero-function is triangle-free.
    Hence, when measuring the last qubit of $\ket{\Psi}$ in the computational basis, outcome $1$ is observed with probability $\geq \varepsilon$.
    Therefore, in any iteration $i$, the first step in our sketched procedure above will succeed and produce some $y_i$ with $f(y_i)=1$ with probability $\geq 1 - \tilde{\delta}$.
    We condition on this high-probability event $E_1$ for the rest of the soundness analysis.
    By an analogous reasoning, the assumption of \Cref{lemma:function-value-subset-state-preparation} is satisfied in the scenario of steps 2 and 3---with $S=\{0,1\}^n$, $\eta=\varepsilon$ and the function either given directly by $f$ or by $f(\cdot+y_i)$---so in any iteration $i$, the second and third step each will succeed with probability $\geq 1-\tilde{\delta}$. We now further condition on these success events $E_2$ and $E_3$.
    At this point, by the same reasoning as in the completeness case, we know that the fourth step, with success probability $\geq 1-2\delta/5$ overall, produces estimates $\hat{\mu}_i$ such that 
    \begin{align*}
        \left\lvert \frac{1}{m}\sum_{i=1}^m \hat{\mu}_i -  \Pr_{x,y\sim\{0,1\}^n}[f(x)=f(y)=f(x+y)=1]\right\rvert
        \leq \tilde{\varepsilon}/3\, .
    \end{align*}
    By \Cref{theorem:classical-soundness-triangle-freeness-testing}, since we assumed $f$ to be $\varepsilon$-far from triangle-free, we have $\Pr_{x,y\sim\{0,1\}^n}[f(x)=f(y)=f(x+y)=1]\geq \tilde{\varepsilon}$.
    Thus, by the first equality from \Cref{eq:tirangle-freeness-proof-intermediate}, we have $$\mathbb{E}_{y\sim\{0,1\}^n: f(y)=1}\left[\Pr_{x\sim\{0,1\}^n}[f(x)=1=f(x+y)]\right]\geq \tilde{\varepsilon}\, .$$
    So, the above implies the inequality $\frac{1}{m}\sum_{i=1}^m \geq 2\tilde{\varepsilon}/3$.
    Hence, the tester will in this case correctly output ``$\varepsilon$-far from triangle-free''.
    A final union bound shows that this occurs with probability $\geq 1-\delta$, which proves soundness.

    The quantum sample complexity of a single iteration is given by $2\cdot 2\left\lceil 162 \ln(6/\tilde{\delta}) \left(6/\tilde{\varepsilon}\right)^{4}\right\rceil\cdot\left\lceil\frac{\ln(2\left\lceil 162 \ln(6/\tilde{\delta}) \left(6/\tilde{\varepsilon}\right)^{4}\right\rceil / \tilde{\delta})}{\varepsilon}\right\rceil + 2\left\lceil 162 \ln(6/\tilde{\delta}) \left(6/\tilde{\varepsilon}\right)^{4}\right\rceil\leq \tilde{\mathcal{O}}\left(\frac{\ln(1/\tilde{\delta})}{\tilde{\varepsilon}^4}\right)$. Thus, the overall quantum sample complexity is $\leq m\cdot \tilde{\mathcal{O}}\left(\frac{\ln(1/\tilde{\delta})}{\tilde{\varepsilon}^4}\right)\leq \tilde{\mathcal{O}}\left(\frac{\ln^2(1/\tilde{\delta})}{\tilde{\varepsilon}^6}\right)$. Plugging in the chosen values for $\tilde{\delta}$ and $\tilde{\varepsilon}$ yields an upper bound of $\tilde{\mathcal{O}}\left(\ln^2(1/\delta)\left(\mathsf{Tower}\left(C\cdot \left\lceil \log\left(\frac{1}{\varepsilon}\right)\right\rceil\right)\right)^{6}\right)$ on the number of quantum copies used by the tester. To achieve the claimed linear dependence on $\ln(1/\delta)$, one can simply run the protocol described above for a constant confidence parameter (say, $\delta= 1/3$), and then amplify the success probability through majority votes. 
    
    Finally, we have to argue that the tester is quantumly computationally efficient. This, however, follows immediately from the efficiency of the procedures from \Cref{lemma:function-value-subset-state-preparation} and \Cref{corollary:two-wise-intersection-estimation}, and from the fact that every $U_{y_i}$ can be implemented by at most $n$ Pauli-$X$ gates.
\end{proof}

\begin{remark} \label{remark:triangle-freeness-lower-bound}
Classically, passive triangle-freeness testing requires at least $\Omega(2^{n/3})$ samples.
We give a brief proof sketch for this lower bound. Suppose we obtain $q$ samples with inputs drawn i.i.d.~uniformly at random from $\{0, 1\}^n$. 
To witness a triangle, we must observe a linearly dependent triple $(x, y, x+y)$.
For any distinct pair $(x, y)$, the probability that $x+y$ also appears among the remaining $q-2$ samples is at most $\frac{q-2}{2^n}$.
A union bound over all possible pairs $(x, y)$ gives us an upper bound of $\frac{q^3}{2^n}$ on the probability that any triple is linearly dependent.
Therefore, we need at least $\Omega(2^{\frac{n}{3}})$ samples for a constant probability of witnessing a triangle (which is required for triangle-freeness testing). 
Just like in symmetry testing, there is an unbounded separation between a classical $n$-dependent and a quantum $n$-independent passive testing sample complexity.
    
\end{remark}

\section{Useful facts}\label{sec:useful-facts}

In this appendix we collect some simple lemmas that are used as subroutines in the paper. 
First, we make a simple observation about the possibility of post-selecting on a desired function value in a subset function state.

\begin{lemma}\label{lemma:function-value-subset-state-preparation}
    Let $m\in\mathbb{N}$, $S\subseteq\{0,1\}^n$, $f:\{0,1\}^n\to\{0,1\}$, $b\in\{0,1\}$, $\eta\in (0,1]$, and $\delta\in (0,1)$.
    Assume that $\Pr_{x\sim S}[f(x)=b]\geq \eta$.
    There is an efficient quantum algorithm that given $m\lceil\frac{\ln(m/\delta)}{\eta}\rceil$ many copies of the state $\ket{\Psi_{S,f}} = \frac{1}{\sqrt{|S|}}\sum_{x\in S}\ket{x,f(x)}$, outputs, with success probability $\geq 1-\delta$, at least $m$ many copies of the state $\ket{\Psi_{S,f,b}} \propto\sum_{x\in S: f(x)=b} \ket{x}$.
    Moreover, if the algorithm fails, then the algorithm explicitly outputs $\mathrm{FAIL}$.
\end{lemma}

We note that via standard Chernoff-Hoeffding bounds, one can achieve the same guarantee as in \Cref{lemma:function-value-subset-state-preparation} using  $\max\{\left\lceil2m/\eta \right\rceil, \left\lceil 2\ln(1/\delta)/\eta^2\right\rceil\}$ many copies of the state $\ket{\Psi_{S,f}}$. This improves the $m$-dependence, but in general comes at the cost of a worse $\eta$-dependence.

\begin{proof}
    By a union bound, it suffices to show that $\tilde{m}=\lceil \frac{\ln(1/\tilde{\delta})}{\eta}\rceil$ many copies of $\ket{\Psi_{S,f}}$ suffice to quantumly efficiently obtain one copy of $\ket{\Psi_{S,f,b}}$ with success probability $\geq 1-\tilde{\delta}$.
    So, let's assume $m=1$.
    Also here the procedure is clear: 
    For each copy of $\ket{\Psi_{S,f}}$, measure the last qubit in the computational basis. 
    If the outcome is $b\in\{0,1\}$, then the post-measurement state after discarding the last qubit is $\ket{\Psi_{S,f,b}}$.
    If, after measuring on all the copies, outcome $b$ has never been observed, output FAIL. Otherwise, output one copy one of the post-measurement states from rounds in which outcome $b$ was observed.

    Again, the analysis of the failure probability is simple:
    \begin{align*}
        \Pr[\mathrm{FAIL}] 
        &= \Pr[\textrm{outcome }b \textrm{ never occurs}]\\
        &\leq (1-\eta)^{\tilde{m}}\\
        &\leq \exp(-\eta \tilde{m})\\
        &\leq \tilde{\delta}\, .
    \end{align*}
    Here, we used the assumption $\Pr_{x\sim S}[f(x)=b]\geq \eta$ and our choice of $\tilde{m}$.
\end{proof}

We also require the following standard routine for estimating the overlap of two pure quantum states.

\begin{lemma}\label{lemma:swap-test-overlap-estimation}
    Let $\varepsilon,\delta\in (0,1)$.
    There is an efficient quantum algorithm that, given $\lceil \frac{2\ln(2/\delta)}{\varepsilon^4}\rceil$ many copies of each of two pure quantum states $\ket{\psi}$ and $\ket{\phi}$, outputs, with success probability $\geq 1-\delta$, an $\varepsilon$-accurate estimate (in $[0,1]$) of the (non-squared) overlap $|\braket{\phi|\psi}|$.
\end{lemma}
\begin{proof}
    The procedure is as follows: Let $m = \lceil \frac{2\ln(2/\delta)}{\varepsilon^4}\rceil$ be the number of copies that are available for each of $\ket{\psi}$ and $\ket{\phi}$. For $1\leq i\leq m$, perform a SWAP test between one copy of $\ket{\psi}$ and one copy of $\ket{\phi}$, let the outcome be $\hat{o}_i$.
    Define $\hat{o} = \frac{1}{m}\sum_{i=1}^m \hat{o}_i$ and output the estimate $\hat{\mu} = \sqrt{2(\hat{o} - \frac{1}{2})}$. Let us analyze the success probability of this procedure.

    First, notice that each SWAP test accepts (\textit{i.e.}, outputs $1$) with probability $\frac{1+|\braket{\phi|\psi}|^2}{2}$ \cite{buhrman2001quantum}. Thus, the $\hat{o}_i$ are i.i.d.~$\mathrm{Bernoulli}(\frac{1+|\braket{\phi|\psi}|^2}{2})$ random variables. So, by a standard Chernoff-Hoeffding bound, we have $\left\lvert\hat{o} - \frac{1+|\braket{\phi|\psi}|^2}{2}\right\rvert\leq \varepsilon^2/2$ with probability $\geq 1-2\exp\left(-m\varepsilon^4/2\right)\geq 1-\delta$, by our choice of $m$.
    As $|\sqrt{x}-\sqrt{y}|\leq \sqrt{|x-y|}$ holds for all $x,y\geq 0$, this implies that also $\left\lvert \hat{\mu}-|\braket{\phi|\psi}|\right\rvert = \left\lvert \sqrt{2(\hat{o} - \frac{1}{2})} - \sqrt{2(\frac{1+|\braket{\phi|\psi}|^2}{2} - \frac{1}{2})}\right\rvert \leq \varepsilon$, with probability $\geq 1-\delta$, as desired.
\end{proof}

Finally, using the overlap estimation routine, one can start from a function state and two function subset states to estimate the probability of an input lying in both 
\begin{corollary}\label{corollary:two-wise-intersection-estimation}
    Let $S,S'\subseteq \{0,1\}^n$, $f, f':\{0,1\}^n\to\{0,1\}$, and $b,b'\in\{0,1\}$.
    Let $\varepsilon,\delta\in (0,1)$.
    There is an efficient quantum algorithm that, given $\lceil \frac{162\ln(6/\delta)}{\varepsilon^4}\rceil$ many copies of each of the states $\ket{\Psi}$, $\ket{\Psi'}$, and $2\lceil \frac{162\ln(6/\delta)}{\varepsilon^4}\rceil$ copies of each of the states $\ket{\Psi_{S,f,b}}$ and $\ket{\Psi_{S',f',b'}}$, outputs, with success probability $\geq 1-\delta$, an $\varepsilon$-accurate estimate of $\Pr_{x\sim\{0,1\}^n}[x\in S\cap S', f(x)=b, f'(x)=b']$.
\end{corollary}
\begin{proof}
    The procedure combines the ingredients developed above. To do so, notice the following equalities:
    \small
    \begin{align*}
        \left\lvert\braket{\Psi|\Psi_{S,f,b}}\right\rvert = \braket{\Psi|\Psi_{S,f,b}}
        &= \sqrt{\frac{|S\cap f^{-1}(b)|}{2^n}}\, , \\
        \left\lvert\braket{\Psi'|\Psi_{S',f',b'}}\right\rvert = \braket{\Psi'|\Psi_{S',f',b'}}
        &= \sqrt{\frac{|S'\cap f'^{-1}(b')|}{2^n}}\, , \\
        \left\lvert\braket{\Psi_{S,f,b}|\Psi_{S',f',b'}}\right\rvert = \braket{\Psi_{S,f,b}|\Psi_{S',f',b'}}
        &= \sqrt{\frac{|S\cap f^{-1}(b)\cap S'\cap f'^{-1}(b')|^2}{|S\cap f^{-1}(b)|\cdot |S'\cap f'^{-1}(b')|}}\, , \\
        \frac{|S\cap f^{-1}(b)\cap S'\cap f'^{-1}(b')|}{2^n}
        &= \sqrt{\frac{|S\cap f^{-1}(b)\cap S'\cap f'^{-1}(b')|^2}{|S\cap f^{-1}(b)|\cdot |S'\cap f'^{-1}(b')|}\cdot\frac{|S\cap f^{-1}(b)|}{2^n}\cdot\frac{|S'\cap f'^{-1}(b')|}{2^n}} \, .
    \end{align*}
    \normalsize
    So, we can estimate $\Pr_{x\sim\{0,1\}^n}[x\in S\cap S', f(x)=b, f'(x)=b'] = \frac{|S\cap f^{-1}(b)\cap S'\cap f'^{-1}(b')|}{2^n}$ as follows:
    \begin{enumerate}
        \item Via the procedure in \Cref{lemma:swap-test-overlap-estimation}, use $\lceil \frac{162\ln(6/\delta)}{\varepsilon^4}\rceil$ many copies of each of $\ket{\Psi}$ and $\ket{\Psi_{S,f,b}}$ to output, with success probability $\geq 1-\frac{\delta}{3}$, an $(\varepsilon/3)$-accurate estimate $\hat{\alpha}$ of $|\braket{\Psi|\Psi_{S,f,b}}|$.
        \item Via the procedure in \Cref{lemma:swap-test-overlap-estimation}, use $\lceil \frac{162\ln(6/\delta)}{\varepsilon^4}\rceil$ many copies of each of $\ket{\Psi'}$ and $\ket{\Psi_{S',f',b'}}$ to output, with success probability $\geq 1-\frac{\delta}{3}$, an $(\varepsilon/3)$-accurate estimate $\hat{\alpha}'$ of $|\braket{\Psi'|\Psi_{S',f',b'}}|$.
        \item Via the procedure in \Cref{lemma:swap-test-overlap-estimation}, use $\lceil \frac{162\ln(6/\delta)}{\varepsilon^4}\rceil$ many copies of each of $\ket{\Psi_{S,f,b}}$ and $\ket{\Psi_{S',f',b'}}$ to output, with success probability $\geq 1-\frac{\delta}{3}$, an $(\varepsilon/3)$-accurate estimate $\hat{\beta}$ of $|\braket{\Psi_{S,f,b}|\Psi_{S',f',b'}}|$.
        \item Output the estimate $\hat{\gamma} = \hat{\beta}\hat{\alpha}\hat{\alpha}'$.
    \end{enumerate}
    By a union bound, the probability that Steps 1-3 all succeed is $\geq 1-\delta$. In this case, we have
    \small
    \begin{align*}
        &\left\lvert \hat{\gamma} -  \Pr_{x\sim\{0,1\}^n}[x\in S\cap S', f(x)=b, f'(x)=b']\right\rvert\\
        &= \left\lvert \hat{\beta}\hat{\alpha}\hat{\alpha}' -  \sqrt{\frac{|S\cap f^{-1}(b)\cap S'\cap f'^{-1}(b')|^2}{|S\cap f^{-1}(b)|\cdot |S'\cap f'^{-1}(b')|}\cdot\frac{|S\cap f^{-1}(b)|}{2^n}\cdot\frac{|S'\cap f'^{-1}(b')|}{2^n}}\right\rvert\\
        &\leq \left\lvert\hat{\beta} - \sqrt{\frac{|S\cap f^{-1}(b)\cap S'\cap f'^{-1}(b')|^2}{|S\cap f^{-1}(b)|\cdot |S'\cap f'^{-1}(b')|}}\right\rvert + \left\lvert \hat{\alpha} - \sqrt{\frac{|S\cap f^{-1}(b)|}{2^n}}\right\rvert + \left\lvert \hat{\alpha}' - \sqrt{\frac{|S'\cap f'^{-1}(b')|}{2^n}}\right\rvert\\
        &\leq 3\cdot\frac{\varepsilon}{3}\\
        &= \varepsilon\, .
    \end{align*}
    \normalsize
    Here, the second step used the triangle inequality together with $\hat{\beta}, \hat{\alpha}$, $\hat{\alpha}'$, $\sqrt{\frac{|S\cap f^{-1}(b)\cap S'\cap f'^{-1}(b')|^2}{|S\cap f^{-1}(b)|\cdot |S'\cap f'^{-1}(b')|}}$, $\sqrt{\frac{|S\cap f^{-1}(b)|}{2^n}}$, $\sqrt{\frac{|S'\cap f'^{-1}(b')|}{2^n}}\in [0,1]$.
\end{proof}

\setcounter{secnumdepth}{0}
\defbibheading{head}{\section{References}}
\sloppy
\printbibliography[heading=head]

\end{document}